\documentclass[a4paper,12pt,twoside]{article}

\pdfoutput=1 
\usepackage[utf8]{inputenc}
\usepackage{arxiv}
\usepackage{graphics}
\usepackage{graphicx}
\usepackage{verbatim}
\usepackage{multicol}
\usepackage{mathtools}
\usepackage{booktabs}
\usepackage{subcaption}
\usepackage{algorithm}
\usepackage{algorithmic}
\usepackage{sidecap}
\usepackage{amssymb}
\usepackage{bbm}
\usepackage{array}
\usepackage{amsthm}
\usepackage{hyperref}
\usepackage{cleveref}
\usepackage{epstopdf}
\usepackage{graphicx,epstopdf} %
\usepackage{authblk}

\DeclarePairedDelimiter\abs{\lvert}{\rvert}%
\DeclarePairedDelimiter\norm{\lVert}{\rVert}%
\makeatletter
\let\oldabs\abs
\def\abs{\@ifstar{\oldabs}{\oldabs*}}
\let\oldnorm\norm
\def\norm{\@ifstar{\oldnorm}{\oldnorm*}}
\makeatother

\newcommand*\diff{\mathop{}\!\mathrm{d}}

\DeclareMathOperator*{\argmin}{arg\,min}

\newtheorem{theorem}{Theorem}[section]
\newtheorem{proposition}[theorem]{Proposition}
\newtheorem{conditions}[theorem]{Conditions}
\newtheorem{exmp}[theorem]{Example}
\newtheorem{remark}[theorem]{Remark}

\numberwithin{equation}{section}

\renewcommand{\thefootnote}{\arabic{footnote}}

\ifpdf
\hypersetup{pdfauthor={M. Holden, M. Pereyra, K. Zygalakis},
  pdftitle={Bayesian inference with data-driven image priors encoded by neural networks}
}
\fi

\author[1,2,3]{Matthew Holden}
\author[1,3]{Marcelo Pereyra}
\author[1]{Konstantinos C. Zygalakis}
\affil[1]{School of Mathematical and Computer Sciences, 
Heriot-Watt University, Edinburgh, Scotland}
\affil[2]{School of Mathematics, University of Edinburgh, 
Edinburgh, Scotland}
\affil[3]{Maxwell Institute for Mathematical Sciences, Bayes 
Centre, 47 Potterrow, Edinburgh, Scotland}

\title{Bayesian imaging with data-driven priors encoded by neural networks: theory, methods, and algorithms}

\begin{document}

\maketitle

\begin{abstract}
This paper proposes a new methodology for performing Bayesian inference in imaging inverse problems where the prior knowledge is available in the form of training data. Following the manifold hypothesis and adopting a generative modelling approach, we construct a data-driven prior that is supported on a sub-manifold of the ambient space, which we can learn from the training data by using a variational autoencoder or a generative adversarial network. We establish the existence and well-posedness of the associated posterior distribution and posterior moments under easily verifiable conditions, providing a rigorous underpinning for Bayesian estimators and uncertainty quantification analyses. Bayesian computation is performed by using a parallel tempered version of the preconditioned Crank-Nicolson algorithm on the manifold, which is shown to be ergodic and robust to the non-convex nature of these data-driven models. In addition to point estimators and uncertainty quantification analyses, we derive a model misspecification test to automatically detect situations where the data-driven prior is unreliable, and explain how to identify the dimension of the latent space directly from the training data. The proposed approach is illustrated with a range of experiments with the MNIST dataset, where it outperforms alternative image reconstruction approaches from the state of the art. A model accuracy analysis suggests that the Bayesian probabilities reported by the data-driven models are also remarkably accurate under a frequentist definition of probability.
\end{abstract}

\section{Introduction}
\renewcommand{\thefootnote}{\arabic{footnote}}

Inverse problems are ubiquitous in imaging science.  Canonical examples include image denoising \cite{houdard2018high}, deconvolution \cite{beck2009fast, afonso2010fast}, compressive sensing \cite{lucka2018enhancing} super-resolution \cite{romano2017little}, tomographic reconstruction \cite{afonso2010fast}, source separation \cite{iordache2012total}, fusion \cite{li1995multisensor}, and phase retrieval \cite{elser2018benchmark}.  The main challenge in solving these problems is that they are typically ill-conditioned or ill-posed, and require additional information to be provided in order to obtain a well-posed solution.\footnote{A problem is well-posed in the sense of Hadamard if a solution exists, is unique, and depends continuously on the observed data.}  The standard approach has been to provide this information in the form of a regularisation term, and solve the inverse problem in a variational framework, where the solution is obtained by minimising an energy function composed of a sum of a data-fidelity term and the specified regularisation term \cite{kaipio2006statistical}.   The Bayesian statistical framework, on the other hand, treats the image of interest as a random quantity. Additional information is then included by specifying a marginal distribution for the image, known as the prior distribution.  The observed data are then incorporated through a likelihood function, to obtain a posterior distribution --- the conditional distribution of the image of interest, given the observed data \cite{kaipio2006statistical}.  While the interpretations of these frameworks differ, they are closely related, as the solution which solves a variational energy function will very often maximise a corresponding posterior density function in the Bayesian setting.

For ill-posed problems, the choice of prior has a significant impact on the solution.  The traditional approach has been to define it analytically, as a hand-crafted function, chosen to encourage specific desired properties such as piece-wise regularity, sparsity in some appropriate basis, smoothness, or an expected spectral profile. Special attention is given in the literature to models with an underlying convex geometry, for which maximum-a-posteriori estimation can be formulated as a convex optimisation problem \cite{ChambollePock}. While these regularisers capture some important aspects of the images, they are overly simplistic and misspecified in the sense that they do not accurately describe their probability distribution (from a modelling perspective, these priors would be poor generative models).

To deliver more accurate solutions capturing fine detail, an alternative approach known as end-to-end modelling has been introduced.   These methods take advantage of the fact that, in many applications, there exists an abundance of prior knowledge in the form of training examples.  Harnessing this information can enable the construction of better models resulting in more accurate solutions.  End-to-end models typically take the form of a neural network taking as input the observed data and returning a reconstructed image.   The training process then optimises the network parameters to learn a mapping minimising some loss function evaluated the reconstruction performance on true images from the training data set.  Methods can be task-agnostic \cite{zhang2017beyond, dong2014learning, gharbi2016deep}, or include the degradation model in the network architecture \cite{gregor2010learning, chen2016trainable}.  These models obtain results far superior to classic variational methods, across a variety of image reconstruction tasks.  In image denoising, methods have progressed to the stage where it has been argued that `\textit{removal of zero-mean white additive Gaussian noise from an image is a solved problem in image processing}' \cite{romano2017little}.

Despite this success, the variational and Bayesian frameworks retain a key advantage:  modularity. The clearly defined roles of the data fidelity and regularisation terms make these methods flexible to changes in the forward model.  End-to-end models are highly sensitive to the task for which they have been trained.  This training process can be expensive, and requires a large amount of data, in the form of clean images and corresponding observations corrupted by the forward model, which may not be available in many applications.  A well-chosen prior distribution or regulariser, on the other hand, can be applied to any forward operator. This has led to a lot of recent interest in the idea of applying existing machine learning techniques to learn a prior distribution from a database of clean images \cite{arridge2019solving}.  Examples in this area include learning regularising functionals \cite{kobler2020total, lunz2018adversarial, li2020nett} and priors implicitly defined by the architecture of convolutional neural networks \cite{ulyanov2018deep}.  Plug-and-Play methods \cite{venkatakrishnan2013plug, ryu2019plug} regularise the problem by incorporating existing denoising algorithms into iterative methods designed to solve variational problems. Regularisation by Denoising \cite{romano2017little,Reehorst2019} is another important contribution to the topic, which formalises such approaches in a Bayesian setting and provides an explicit derivation of denoiser-based priors in terms of an image-adaptive Laplacian regularization.

This work takes an alternative approach, defining the prior through a generative model.  The use of a generative prior has previously been explored in \cite{bora2017compressed}, which parametrises the prior using a generative adversarial network, followed by inference by optimisation with respect to the latent variable. This optimisation problem is not convex, however, motivating the work of \cite{gonzalez2019solving}, who use a non-deterministic generative model which enables optimisation in an augmented state space, and a gradient descent which is guaranteed to converge.  Alternative approaches to solving inverse problems using the latent variable of a generative model include \cite{menon2020pulse, marinescu2020bayesian}, which are based on previous works aiming to invert generative models and project images to the corresponding latent space \cite{bojanowski2018optimizing, creswell2018inverting, abdal2019image2stylegan}.

Although motivated through the Bayesian framework, the aforementioned works are variational in spirit, treating the prior as a regularisation term in an optimisation problem and obtaining a single point estimate maximising the posterior probability density.  While the works are promising, many questions about the theoretical properties of the resulting posterior distribution remain unanswered, including well-posedness of the posterior distribution and existence and well-posedness of Bayesian estimators. In addition, the computational algorithms used do not necessarily have convergence guarantees, or may only converge locally.  

In this paper, we take a fully Bayesian approach to the problem, enabling us to answer some of these outstanding questions.  Instead of solving for a single estimated image which minimises an optimization problem, we treat images as random quantities and use Bayes' formula to combine the observed and prior information to obtain a posterior distribution.  Rather than directly condensing this posterior into a single point estimate, as in previous works, we seek to represent the full posterior distribution by Monte Carlo sampling. This approach allows us to compute Bayesian estimators, but also to tackle a number of important questions which are beyond the scope of variational methods.  For example, since data-driven priors can be highly specialised to the specific image classes in the training dataset, one may be concerned that misapplication to unknown images that are not well represented by the training dataset will lead to unreliable results. Adopting a statistical approach, we propose a hypothesis test for identifying such scenarios directly from the observed data, avoiding the misapplication of an incorrect prior.

This approach also provides a means of investigating the accuracy of the prior for the intended dataset.  As previously stated, analytic priors do not usually accurately describe the probability distribution of the images considered. As a result, Bayesian probabilities are not valid in a frequentist sense (i.e., frequencies observed over a large number of repetitions of the experiment). A data-driven prior could hope to accurately represent the true distribution of considered images, in which case the Bayesian and the frequentist probabilities coincide.

Tackling each of these questions requires the evaluation of integrals with respect to the posterior distribution.  In high dimensions, these integrals are analytically intractable and require approximation by Monte Carlo integration by using samples from the posterior distribution. Markov Chain Monte Carlo (MCMC) sampling methods are well-established in the context of imaging problems \cite{kaipio2006statistical, pereyra2015survey}, but efficient MCMC sampling with high-dimensional data-driven priors is highly non-trivial. In this paper, we mitigate this difficulty by learning a prior in the form of a generative model which represents the distribution on a latent manifold of much lower dimensionality. To summarise, in this work, we
\begin{enumerate}
	\item Introduce a methodology to incorporate a data-driven prior in a Bayesian statistical framework, where the prior takes the form of a generative model, and show that easily verified conditions guarantee this defines a well-posed Bayesian inverse problem.
	\item Exploit generative models with inherent dimensionality reduction to construct a posterior distribution which can be sampled efficiently using a robust MCMC algorithm.
	\item Illustrate the power of the proposed methodology by interrogating the posterior distribution in a variety of numerical experiments, including point estimation and comparison with state of the art results; visualisation of uncertainty; detection of out-of-dataset observations; and an analysis of the quality of the prior obtained from the data.
\end{enumerate}

The remainder of the paper is organized as follows. Section \ref{sec:problemstatement} introduces notation and defines the problem.  Section \ref{sec:meth} introduces the generative model used to define the prior and outlines how it can be used in a gradient-free sampling algorithm \cite{cotter2013mcmc}.  In Section \ref{sec:num}, numerical experiments on the MNIST dataset \cite{lecun1998mnist} are presented for three canonical imaging tasks, in which we investigate the questions previously mentioned,  and obtain point estimates superior to those obtained using an optimisation approach.  We conclude in Section \ref{sec:discus} with a discussion about future directions of this research. 

\section{Problem Statement} \label{sec:problemstatement}

We consider imaging problems involving an unknown image $x \in \mathbb{R}^d$ and some observed data $y \in \mathbb{C}^p$, related to $x$ through {a} statistical model with likelihood function $p(y|x)$. We are particularly interested in problems where the estimation of $x$ from $y$ is ill-posed or ill-conditioned ({i.e., either the problem does not admit a unique solution that changes continuously with $y$, or there exists a unique solution but it is not stable w.r.t. small perturbations in $y$). To illustrate throughout, we will consider problems of the form $y = Ax + w$ with $w \sim \mathcal{N}(0,\sigma^2 \mathbb{I}_p)$ and $\sigma>0$ where the observation operator $A \in  \mathbb{C}^{d \times p}$ is rank deficient, or problems where $A^\top A$ is full rank but has a poor {condition} number}. It is important to emphasize that these are simply illustrative examples, and the methodology we propose is not restricted to problems of this form, and problems with alternative noise distributions, or a nonlinear forward operator, for example, could also be tackled.

The ill-posedness of the estimation problem implies that additional information is required in order to reduce the uncertainty about $x$ and deliver meaningful solutions. Within the Bayesian framework, this is achieved by exploiting the marginal distribution of $x$, the so-called prior distribution $p(x)$, which is combined with the likelihood to derive the posterior distribution
\begin{equation*}
p(x|y)=\frac{p(y|x)p(x)}{\int_{\mathbb{R}^{d}} p(y|x)p(x)dx}\, .
\end{equation*}

Of course, the true marginal distribution of $x$ is unknown, so approximations or operational priors are used instead, as we have already discussed. Most Bayesian imaging methods use operational priors that are specified analytically. Analytical priors usually promote specific expected or desired properties about $x$, such as sparsity on some appropriate basis or dictionary, smoothness, or piecewise regularity. Special attention is given in the literature to models with density functions that are log-concave, as this allows the use of computationally efficient Bayesian computation algorithms and greatly simplifies their theoretical convergence analysis \cite{ChambollePock, durmus2018efficient}. 

In this paper we study a Bayesian methodology for settings in which prior knowledge is available in the form of a collection of training examples $\{x^\prime_i\}_{i=1}^n$, which we consider to be independent samples from $p(x)$. We are particularly interested in questions related to performing Bayesian analysis and computation in this context, where $p(x)$ is only known through the representatives $\{x^\prime_i\}_{i=1}^n$. The posterior distribution of interest, henceforth denoted by $\pi$, is given by 
\begin{equation}\label{eqn:piDist}
\pi (x) \triangleq p(x|y,\{x^\prime_i\}_{i=1}^n) =\frac{p(y|x)p(x|\{x^\prime_i\}_{i=1}^n)}{\int_{\mathbb{R}^{d}} p(y|x)p(x|\{x^\prime_i\}_{i=1}^n)dx}\, ,
\end{equation}
where the prior $p(x|\{x^\prime_i\}_{i=1}^n)$ explicitly models that the information available about the marginal of $x$ comes from the observed training examples $\{x^\prime_i\}_{i=1}^n$.

Performing Bayesian computation for $\pi$ is highly non-trivial given the dimensionality involved in imaging problems, and because a construction of $p(x|\{x^\prime_i\}_{i=1}^n)$ that accurately captures the information in $\{x^\prime_i\}_{i=1}^n$ will not naturally lead to a posterior distribution which is amenable to Bayesian computation.
In a manner akin to \cite{gonzalez2019solving}, in this paper we seek to address this difficulty by exploiting the so-called \emph{manifold hypothesis} \cite{fefferman2016testing} that plays a central role in modern imaging sciences; \emph{i.e.}, that natural images take values in the neighbourhood of a low-dimensional sub-manifold of the ambient space. Operating on this manifold of dramatically reduced dimensionality will simultaneously allow us to effectively regularise the estimation problem and to perform computations efficiently.

\section{Proposed Bayesian methodology} \label{sec:meth}
The methodology we propose achieves this dimensionality reduction by defining the prior through a generative model, as in \cite{bora2017compressed, gonzalez2019solving, marinescu2020bayesian}.  In this section, we show that such a prior leads to a well-posed Bayesian inverse problem, under easily-satisfiable conditions.  While posterior well-posedness is not often addressed in the imaging context for non-convex problems, it is a crucial property which ensures that point estimates and uncertainty quantification analyses obtained from the posterior are stable with respect to the data.  We then take a fully Bayesian approach to the problem, showing how Markov Chain Monte Carlo can be used to generate samples from the posterior distribution we obtain.  We further establish that the resulting Markov Chain is ergodic, thus guaranteeing that large samples obtained through the MCMC algorithm represent the statistical properties of the posterior distribution.

\subsection{Defining a Generative Prior} \label{subsec:vae}
We seek to construct a prior $p(x|\{x^\prime_i\}_{i=1}^n)$ that is supported on an $m$-dimensional sub-manifold of $\mathbb{R}^d$ and which accurately models the empirical distribution of the training samples $\{x^\prime_i\}_{i=1}^n$. To achieve this, we first introduce a latent variable $z \in \mathbb{R}^{m}$ and assign it a tractable marginal distribution $p(z)$ on $\mathbb{R}^m$, for example $z \sim \mathcal{N}(0,\mathbb{I}_m)$. We then use the samples $\{x^\prime_i\}_{i=1}^n$ to learn a mapping $\mu: \mathbb{R}^{m} \mapsto \mathbb{R}^{d}$ such that the $\mu$-pushfoward measure of $z$ on $\mathbb{R}^d$ is close to $\{x^\prime_i\}_{i=1}^n$, as measured by some appropriate goodness-of-fit criterion. The distribution of interest $p(x|\{x^\prime_i\}_{i=1}^n)$ is then obtained by marginalising $z$ from an augmented distribution $p(x,z|\{x^\prime_i\}_{i=1}^n) = p(x|z)p(z)$, where $p(x|z)$ depends on $\mu$, which contains the information from $\{x^\prime_i\}_{i=1}^n$. For computational efficiency, we focus on the case where $\mu$ is deterministic, the distributions $p(x,z|\{x^\prime_i\}_{i=1}^n)$ and $p(x|\{x^\prime_i\}_{i=1}^n)$ are degenerate, and $p(x|\{x^\prime_i\}_{i=1}^n)$ is supported on the $m$-dimensional sub-manifold of $\mathbb{R}^d$ which is specified by $\mu$.  (See Figure \ref{fig:pushforwardsamples} for an illustration of this approach.)

The accurate estimation of the mapping $\mu$ from $\{x^\prime_i\}_{i=1}^n$ is a long-standing problem in the generative modelling machine learning literature. Recently, strategies where $\mu$ is represented by a deep neural network have received a lot of attention, with adversarial networks (GANs) \cite{goodfellow2014generative} and variational autoencoders (VAEs) \cite{kingma2013auto} being two prominent examples. In this paper, we chose to use a VAE because, in addition to an estimate of $\mu$, VAEs also provide an estimate of a pseudo-inverse of $\mu$ that is useful for identifying the manifold dimension $m$ and for performing uncertainty quantification analyses. However, the proposed methodology is presented in a general form and could equally be implemented using a GAN, or other models to generate images by sampling a latent variable with a specified distribution in a lower-dimensional space.

To repeat, VAEs provide a specific way of estimating $\mu$ to construct the augmented distribution $p(x,z|\{x^\prime_i\}_{i=1}^n)$, which admits $p(x|\{x^\prime_i\}_{i=1}^n)$ as marginal. A key ingredient in the VAE construction is the approximation of the unknown conditional distribution $p(x|z,\{x^\prime_i\}_{i=1}^n)$ by a parametric approximation $p_\theta (x|z)$ based on a neural network with weights or coefficients $\theta$. This distribution is known as the decoder because it relates the latent (or codified) image representation $z$ to the ambient space image representation $x$ (notice that the decoder is stochastic and hence the image $x$ decoded from any $z$ is a random quantity). Following a maximum likelihood estimation approach, the value of $\theta$ is set to maximise the marginal likelihood of $\{x^\prime_i\}_{i=1}^n$, which under an independence assumption can be expressed as $p_\theta(x^\prime_1,\ldots,x^\prime_n) = \prod_{i=1}^n p_\theta(x^\prime_i)$, where $p_\theta(x^\prime_i) = \int p_\theta(x^\prime_i|z) p(z)  \, dz$.

The above maximum likelihood estimation problem is computationally intractable so a variational lower bound on the log-likelihood is optimised instead.  For this, an approximation $q_\phi(z|x)$ of $p(z|x,\{x^\prime_i\}_{i=1}^n)$, is introduced, also parametrised by a deep neural network known as the encoder. Applying Jensen's inequality to
\[
p_\theta(x^\prime_i) = \int p(z) \frac{p_\theta(x^\prime_i|z)}{q_\phi(z|x^\prime_i)} q_\phi(z|x^\prime_i) \, dz,
\]
one obtains the following {variational lower bound} on the log-likelihood
\begin{equation} \label{ELBO}
\log p_\theta(x^\prime_i) \geq \mathcal{L}(\theta, \phi;x^\prime_i) \triangleq \mathbb{E}_{q_\phi} [\log p_\theta(x^\prime_i|z)] - D_{KL} (q_\phi(z|x^\prime_i) || p_Z(z)),
\end{equation}
which is used in VAEs as a surrogate of $\log p_\theta(x^\prime_1,\ldots,x^\prime_n)$ in order to estimate $\theta$ (and $\phi$). A standard choice in VAE modelling, which we adopt in this paper, is as follows. 

\begin{enumerate}
    \item Let $p_\mathcal{N}(z;\mu,\Sigma)$ denote the density of the multivariate $\mathcal{N}(\mu, \Sigma)$ distribution.
    \item Set $p(z) = p_\mathcal{N}(z;0,\mathbb{I}_m)$.
    \item Set $q_\phi(z|x) = p_\mathcal{N}(z;\mu_\phi(x), \sigma_\phi^2(x) \mathbb{I}_m)$, where $\mu_\phi : \mathbb{R}^d \mapsto \mathbb{R}^m$ is the neural network specifying the expectation of the (stochastic) encoding step, and $\sigma^2_\phi : \mathbb{R}^d \mapsto \mathbb{R}_+$ controls the encoding variance.
    \item Set $p_\theta(x|z) = p_\mathcal{N}(x;\mu_\theta(z), \sigma_\theta^2(z) \mathbb{I}_d)$, where $\mu_\theta : \mathbb{R}^m \mapsto \mathbb{R}^d$ is the neural network specifying the expectation of the (stochastic) decoding step, and $\sigma^2_\theta : \mathbb{R}^m \mapsto \mathbb{R}_+$ controls the decoding variance.
\end{enumerate}

The choice of Gaussian encoders and decoders simplifies the estimation of $\theta$ and $\phi$, as in this case the KL divergence in \eqref{ELBO} can be evaluated analytically. The expectation in \eqref{ELBO} is not tractable but can be approximated by Monte Carlo integration by using a subset or mini-batch of $x^\prime_1,\ldots,x^\prime_n$. The weights $\theta$ and $\phi$ specifying the encoding and decoding networks are then estimated by using off-the-shelf stochastic gradient optimisation techniques, such as Adam \cite{kingma2014adam}. We refer the reader to \cite{kingma2013auto} for an excellent introduction to VAEs and more details about their implementation.

\begin{figure}[t]
	\centering
	\begin{subfigure}{.48\textwidth}
		\centering
		\includegraphics[width=\textwidth]{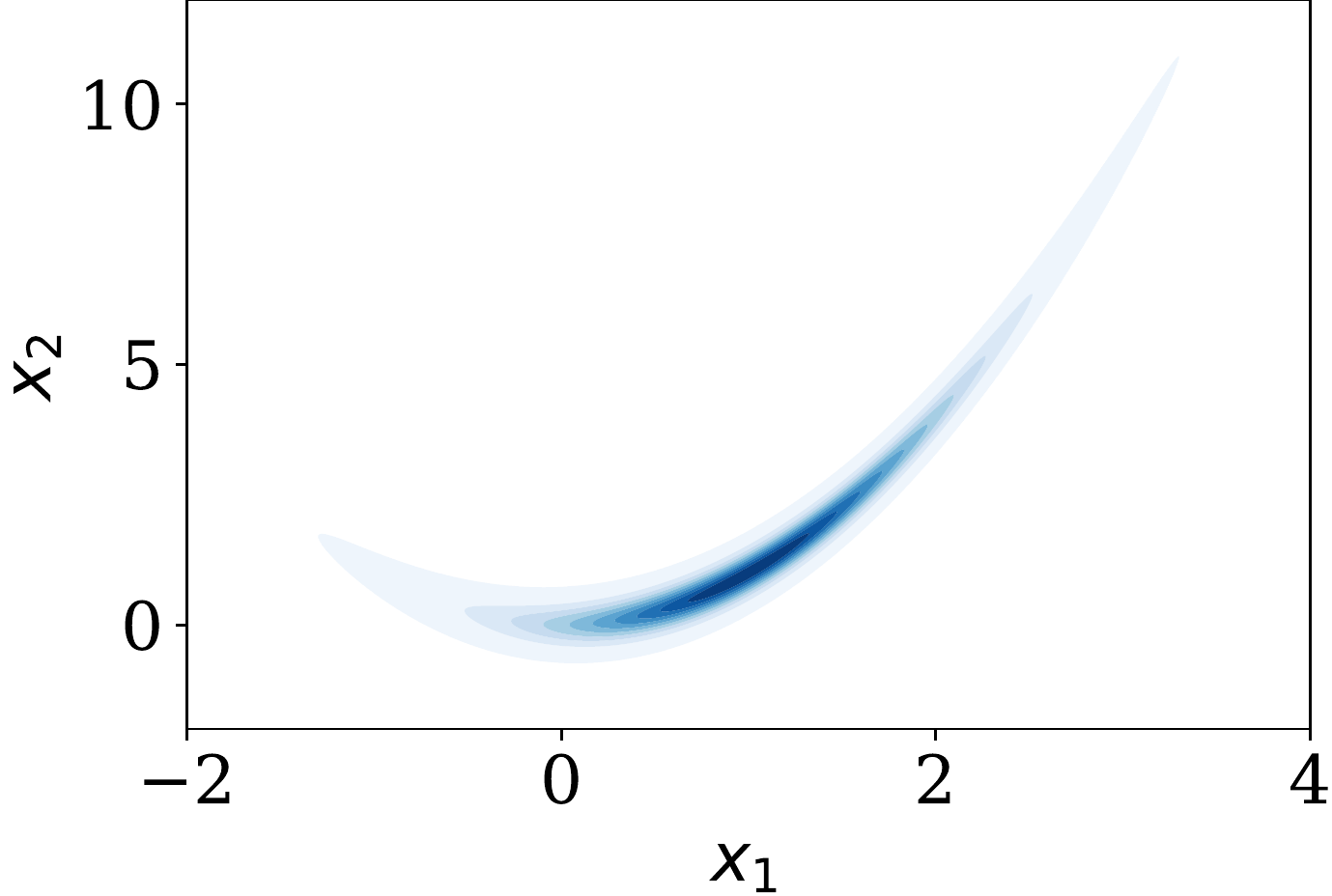}
		\caption{Probability Density}
		\label{subfig:rosendensity}
	\end{subfigure}	
	\hfill
	\begin{subfigure}{.48\textwidth}
		\centering
		\includegraphics[width=\textwidth]{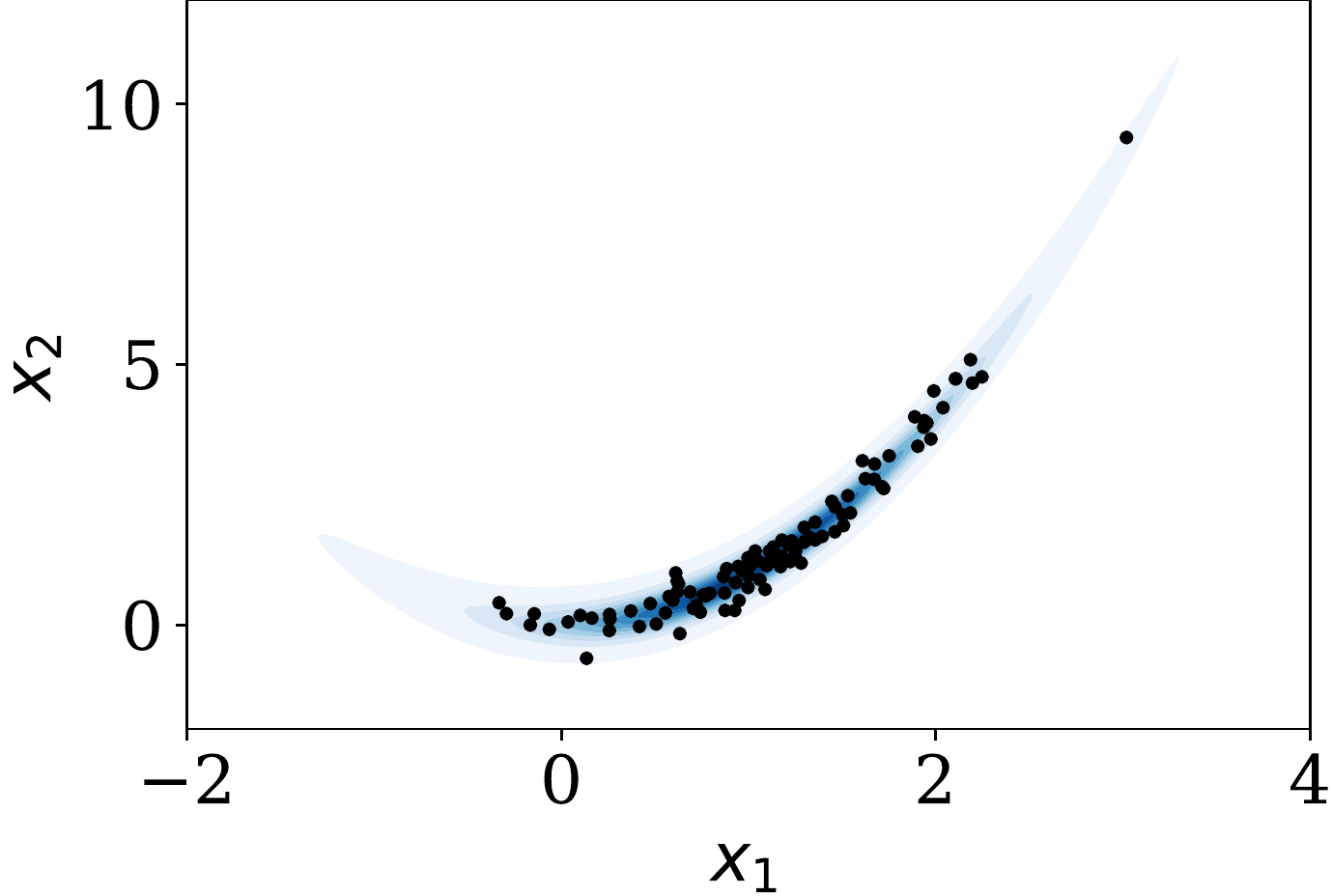}
		\caption{True Samples}
		\label{subfig:truesamples}
	\end{subfigure}	
	\hfill
	\begin{subfigure}{.48\textwidth}
		\centering
		\includegraphics[width=\textwidth]{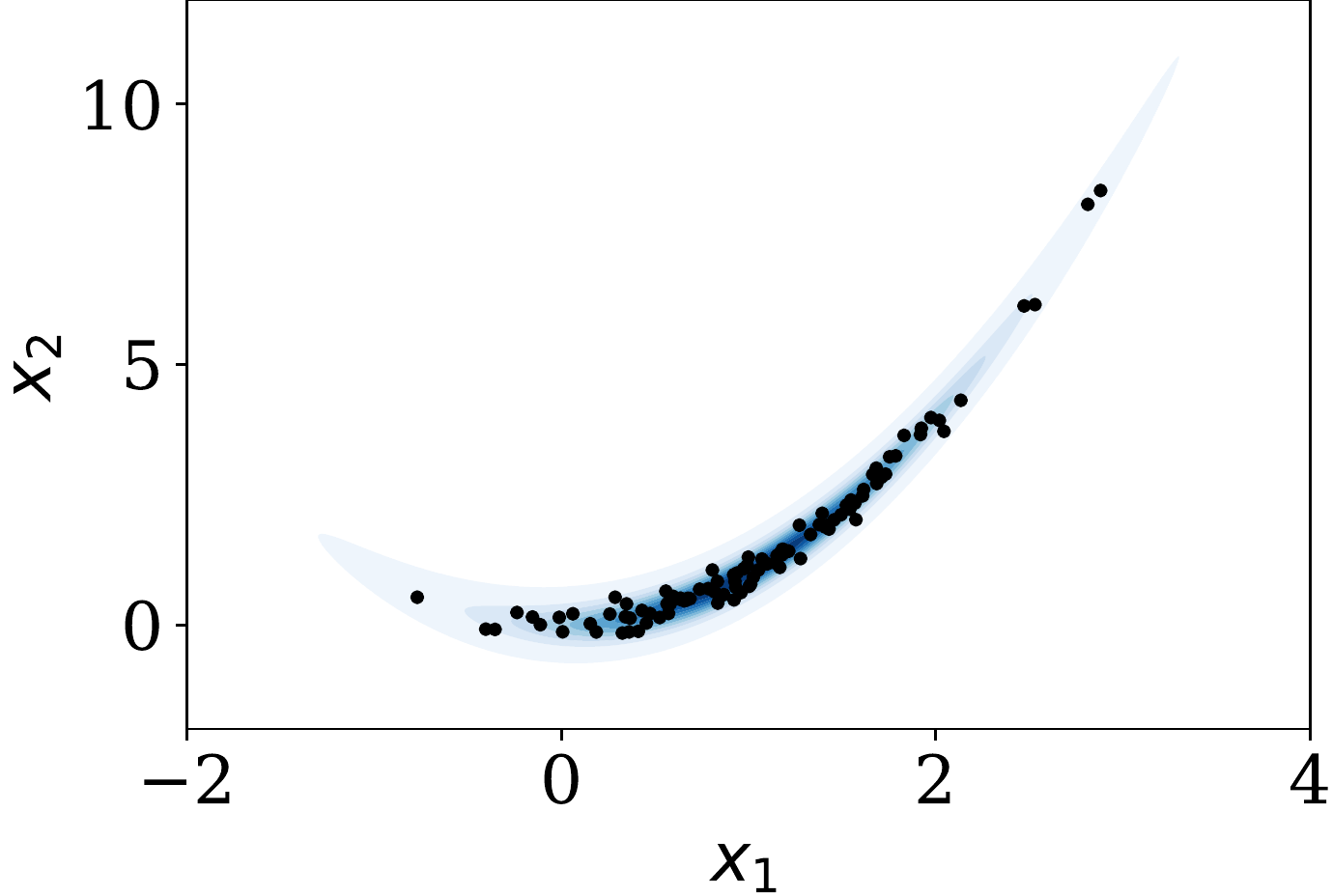}
		\caption{Stochastic VAE samples}
		\label{subfig:stochastic_samples}
	\end{subfigure}
	\hfill
	\begin{subfigure}{.48\textwidth}
		\centering
		\includegraphics[width=\textwidth]{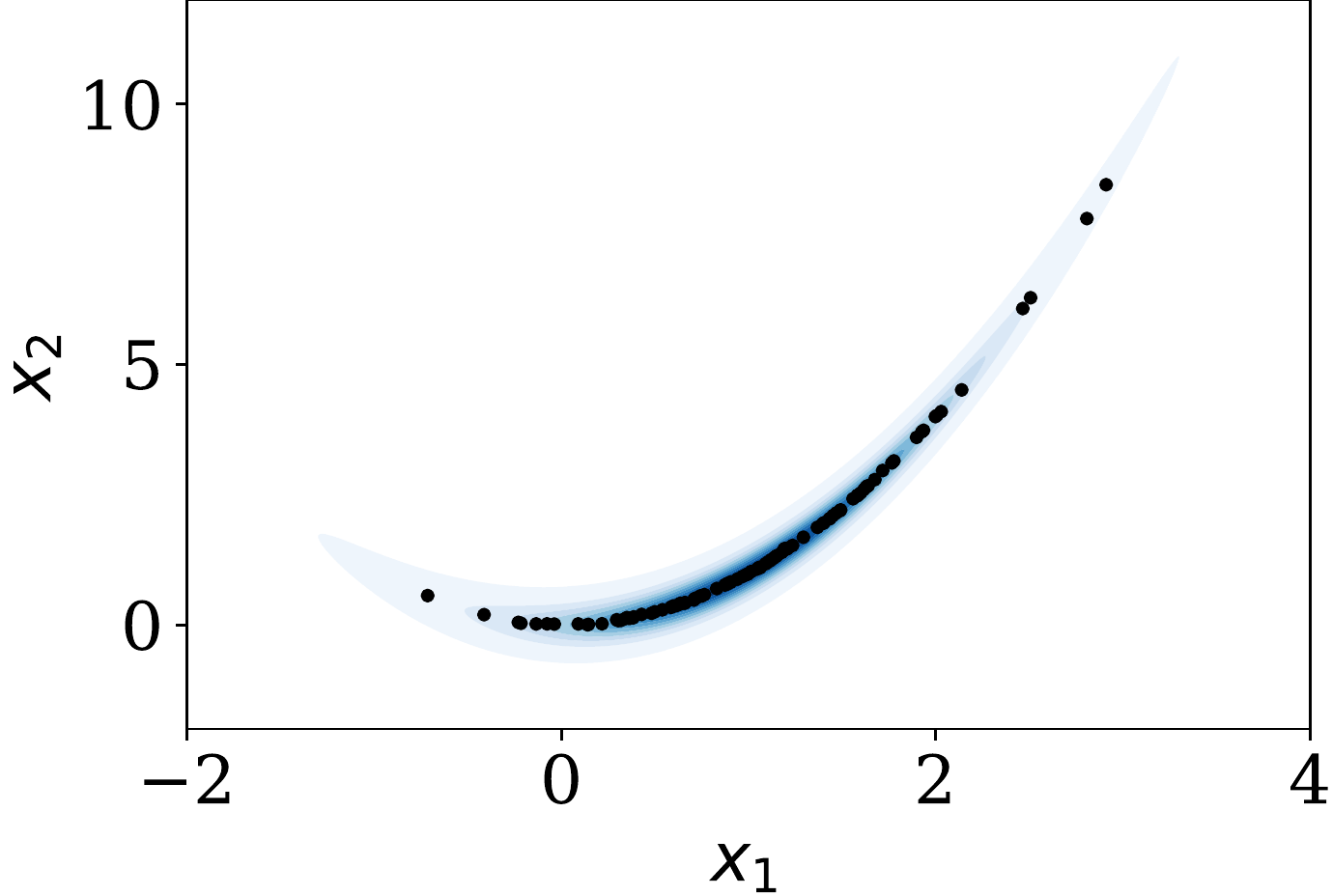}
		\caption{Deterministic VAE samples}
		\label{subfig:deterministic_samples}
	\end{subfigure}
	\caption{Rosenbrock Distribution: (a) Probability density function, as defined in Example \ref{ex:rosenbrock}. (b) 100 samples generated according to true distribution $p(x_1,x_2)$.  (c) 100 samples generated by stochastic VAE with 1D latent space, using the two-step process $z \sim p(z) = p_\mathcal{N}(z;0,1)$ and $(x_1,x_2) \sim p_\theta((x_1,x_2)|z) = p_\mathcal{N}((x_1,x_2);\mu_\theta(z), \sigma_\theta^2(z) \mathbb{I}_2)$. (d) 100 samples generated using $\mu_\theta$-pushforward measure of $z \sim p(z) = p_\mathcal{N}(z;0,1)$ under the deterministic mapping $\mu_\theta$ learnt by the VAE.}
	\label{fig:pushforwardsamples}
\end{figure}

Lastly, note that marginalising $z$ in the VAE construction leads to the conditional distribution
$
p(x|x^\prime_1,\ldots,x^\prime_n) = \int p_{\mathcal{N}}(x|\mu_\theta(z),\sigma_\theta^2(z)\mathbb{I}_d) p_{\mathcal{N}}(z|0,\mathbb{I}_m) \textrm{d}z\,.
$
This distribution is supported on $\mathbb{R}^d$ and concentrates probability mass around the image of mapping $\mu_\theta : \mathbb{R}^m \rightarrow \mathbb{R}^d$, where the term $\sigma_\theta^2$ controls the degree of concentration. Operating directly with this distribution as prior would be very computationally challenging, as the associated posterior
is high-dimensional, nearly degenerate, and highly anisotropic. Instead, to simplify Bayesian computation and analysis tasks, we choose to collapse the stochastic decoder into a deterministic decoder by setting $x = \mu_\theta(z)$ in order to obtain a deterministic mapping from $z$ to $x$ and a degenerate prior on $x$ that is supported on the manifold (i.e., we set $\mu \xleftarrow[]{} \mu_\theta$). This allows one to operate directly with $z|y,x^\prime_1,\ldots,x^\prime_n$ to benefit from dramatically reduced dimensionality; we then define $x|y,x^\prime_1,\ldots,x^\prime_n$ as the $\mu_\theta$-pushforward measure of $z|y,x^\prime_1,\ldots,x^\prime_n$.

\begin{exmp} \label{ex:rosenbrock}
	As a simple example to illustrate this method, consider the 2D Rosenbrock distribution \cite{rosenbrock1960automatic}, defined as
	\[
	p(x_1,x_2) \propto \exp \left\{ - 8(x_2 - x_1^2)^2 + (1 - x_1)^2\right\}.
	\]
	This joint distribution in $x_1$ and $x_2$ is typically quite challenging for MCMC algorithms to sample efficiently, as the mass is concentrated on a narrow curved ridge (Figure \ref{subfig:rosendensity}).  As a result, it is often used as a benchmark problem when testing algorithms.  For an algorithm with the information that $(x_1, x_2)$ lie close to the parabola $x_2 = x_1^2$, sampling is much simpler.  To illustrate how a generative model can capture this type of relationship from training examples, a VAE was trained on $10^5$ samples from the Rosenbrock distribution, using a 1-dimensional latent variable with standard normal prior $p(z) = p_\mathcal{N}(z;0,1)$.  As illustrated in Figure \ref{fig:pushforwardsamples}, the VAE learns the parabolic relationship between $x_1$ and $x_2$.  The figure also illustrates the effect of collapsing the stochastic decoder into a deterministic one, as the deterministically generated samples are restricted to the 1D curve learnt by the VAE.
\end{exmp}
\subsection{Posterior distribution in the latent variable} \label{subsec:posteriorlatent} 

Having estimated the mapping $\mu_\theta$ from the training samples $\{x^\prime_i\}_{i=1}^n$, a posterior distribution in the latent variable can be derived.  The prior is simply the tractable marginal $p(z)$ assigned to the latent variable, while the likelihood is defined as $p(y|z) = p(y|x=\mu_\theta(z))$.  The posterior distribution in the latent variable $z$ is then given by 
\begin{equation} \label{eq:latentpost}
p(z|y,\{x^\prime_i\}_{i=1}^n) \propto p(y|x=\mu_\theta(z)) p(z).
\end{equation}
For ease of notation, we henceforth denote this posterior on the latent space as simply $p(z|y)$.

\begin{exmp} \label{ex:gaussianmodel}
Consider the example $y=Ax + w$ where $w \sim N(0, \sigma_y^2\mathbb{I}_p)$. If $\mu_\theta$ is a generative model for $x$ with prior  $z \sim N(0, \mathbb{I}_m)$, the posterior can be written
\begin{equation} \label{eq:gaussianpost}
p(z|y) \propto \exp \left\{-\frac{(y - A \mu_\theta(z))^\dagger (y - A \mu_\theta(z))}{2 \sigma_y^2}  - \frac{z^\top z}{2} \right\} \, .
\end{equation}
\end{exmp}
Since the latent space is of lower dimension than the pixel space, we cannot transform the latent posterior to give an explicitly evaluable posterior density, $p(x|y)$, in the space of images.  Indeed, defining the distribution of $x|y$ as the pushforward of $z|y$ under $\mu_\theta$ we obtain a posterior in the image space which does not admit a density with respect to the Lebesgue measure on $\mathbb{R}^d$.  Nevertheless, it is possible to sample from $x|y$, since $x$ is assumed to be generated deterministically from the latent variable $z$ -- simply generate samples $z_i \sim p(z|y)$ and transform $x_i = \mu_\theta(z_i)$.  These samples can then be used to perform advanced Bayesian analysis, such as uncertainty quantification and model selection \cite{durmus2018efficient} along with point estimate image reconstructions.
\subsubsection{Well-posedness of posterior distribution}

Although the posterior in the image space does not have a density with respect to the Lebesgue measure on $\mathbb{R}^d$, it is still possible to guarantee that the posterior distribution of $z|y$ and its $\mu_\theta$-pushforward in the image space $x|y$ are well-posed. A Bayesian inverse problem is well-posed if the posterior distribution exists, is unique, and depends continuously on the observation $y$, where continuity is defined with respect to a chosen metric on the space of probability measures.  Well-posedness in a particular metric can guarantee continuity in quantities of interest calculated from the posterior, thereby ensuring that Bayesian estimators and uncertainty quantification analyses are themselves well-posed in the sense of Hadamard.  

In this section we give conditions guaranteeing that the posterior distribution obtained using the data-driven prior is well-posed in the Prokhorov, total variation, Hellinger and Wasserstein distances.  These different notions of well-posedness are discussed in detail in \cite{latz2020well}, but the guarantees resulting from them can be summarised as follows.  Weak  well-posedness (well-posedness in the Prokhorov metric) implies continuity of posterior expectations of bounded, continuous quantities of interest.  This is particularly relevant when calculating posterior moments, such as the mean and variance of the distribution.  Total variation and Hellinger well-posedness extend this guarantee to any  bounded  quantity of interest, so are relevant when discussing statements of the form $\mathbb{P}_{x|y}(x \in R)$.  Finally, continuity in the Wasserstein distance has been used to prove convergence and stability of MCMC algorithms \cite{rudolf2018perturbation}.

If $z$ and $y$ are random variables and $p(z)$ denotes the prior distribution on $z$, then the following conditions on the likelihood function $p(y|z)$ are sufficient for the Bayesian inverse problem of obtaining the posterior probability distribution to be well-posed.
\begin{conditions} \label{conditions:wellposed}
	\ 
	\begin{itemize}
		\item[(C1)]  $p(\cdot|z)$ is a strictly positive probability density function,
		\item[(C2)]  $\int \abs{p(y|z)}\, p(z) \, dz < \infty$,
		\item[(C3)]  There exists a function $g(z)$ such that $\int \abs{g(z)}\, p(z) \, dz < \infty$, and $p(y|z) \leq g(z)$ for all $y$, and
		\item[(C4)]  $p(\cdot|z)$ is continuous.
		\item[(C5)]  $\int \norm{z}^p p(z) \, dz < \infty$ for some $p \in [1, \infty)$.
		\item[(C6)]  There exists $c \in (0, \infty)$ such that $p(y|z) \leq c$ for all $y$ and with probability one in $z$.
	\end{itemize}
\end{conditions}

\begin{theorem}[{\cite[Theorems 3.6, 3.12]{latz2020well}}]\label{theorem_wellposed}
	If $p(y|x)$ satisfies conditions (C1-C4) then the Bayesian inverse problem is weakly, Hellinger, and total variation well-posed.  If, additionally, conditions (C5) and (C6) are satisfied, then the problem is Wasserstein(p) well-posed.
\end{theorem}

\begin{remark}\label{example3p4}
	Let the measurement error in $y$ be modelled as additive Gaussian noise, 
	\[
	p(y|z) = \det (2\pi\Sigma)^{-1/2} \exp \left( -\frac{1}{2} \norm{\Sigma^{-1/2} (y - A \mu_\theta(z))}^2\right),
	\]
	where $A \in  \mathbb{C}^{n \times p}$ is a linear forward operator.  Then the Bayesian inverse problem is weakly, Hellinger, and total variation well-posed for any prior $p(z)$ defining a probability density. Furthermore, for any such prior satisfying (C5), the problem is Wasserstein(p) well-posed.  In particular, for the latent posterior model of Example \ref{ex:gaussianmodel}, the Bayesian inverse problem is well-posed.  A more general version of this theorem, extended to Radon-Nikodym derivatives, as outlined in \cite{latz2020well}, can be used to show that the problem in the image space, with prior given by the $\mu_\theta$-pushforward measure of the latent prior, is also well-posed.
\end{remark}

\subsubsection{Existence of posterior moments}

One of the most common ways to summarise the posterior distribution is through its moments.  While the previous section established conditions guaranteeing the well-posedness of bounded quantities, further work is required to establish that the moments of the distribution exist. In this section we show that all posterior moments exist for our model as long as the prior moments exist and $\mu_\theta$ is Lipschitz.  The condition on the former holds trivially when the prior is defined to be Gaussian, while when $\mu_\theta$ is defined using a neural network, the Lipschitz requirement simply states that the weights and biases must be finite and the activation functions Lipschitz.

\begin{proposition}\label{momentexistence}
	Assume $\mu_\theta: \mathbb{R}^m \to \mathbb{R}^d$ is L-Lipschitz, the likelihood $p(y|z)$ defines a probability density function, and the prior on the latent space has finite i\textsuperscript{th} moment $\mathbb{E}_p(\abs{z}^i)$ for $i=1,...,k$.  Then the posterior k\textsuperscript{th} moment $\mathbb{E}_\pi(\abs{x}^k)$ exists for almost all observations $y$.  That is,
	\[
	\mathbb{P}_y[E_\pi(\abs{\mu_\theta(z)}^k) < \infty] = 1.
	\]
\end{proposition}
\begin{proof}
	By Lipschitz continuity $\abs{\mu_\theta(z)} \leq \abs{\mu_\theta(0)} + L \norm{z}$, so it suffices to show $\mathbb{P}[E_\pi(\abs{z}^i) < \infty] = 1$ for all $i=1,...,k$. This can be concluded immediately from Fubini's Theorem, since $\mathbb{E} [\mathbb{E}_\pi(\abs{z}^i) ] = \mathbb{E}_p(\abs{z}^i) < \infty $. 
\end{proof}

Together with \Cref{theorem_wellposed} and \Cref{example3p4}, \Cref{momentexistence} states that important Bayesian estimators such as the posterior mean are not only guaranteed to exist, but in fact are guaranteed to be well-posed estimators of $x$.

\subsection{Markov Chain Monte Carlo} \label{subsec:mcmc}

The posterior density (\ref{eq:latentpost}) is known only up to proportionality.  To calculate expectations, Monte Carlo integration must be used.  Typically in imaging applications, the dimensionality of the problem makes efficient sampling difficult unless the posterior (and hence prior) is required to be log-concave. In lower dimensions, this restriction is unnecessary, and posteriors which are not log-concave can be sampled efficiently.   For this, we propose using the preconditioned Crank--Nicolson algorithm \cite{cotter2013mcmc}, which provides a robust baseline for sampling such problems and scales well with increasing dimensionality, being well-defined in infinite dimensions.  Alternative methods such as importance sampling are ill-suited to this problem as, even given the dimensionality reduction, we expect the highly concentrated, multi-modal nature of the latent space to be beyond their scope.  Another option would be to use a gradient-based sampler such as Langevin or Hamiltonian Monte Carlo, but these are not straightforward to apply in this case.  A particular problem is that the maximum step size permitted is limited by the Lipschitz constant of the gradient of the potential.  For a posterior defined by a neural network, this Lipschitz constant may be difficult to calculate and, worse, the potential may not be gradient Lipschitz at all, as would be the case for any network architecture including ReLU nonlinearities.  This would rule out many of the most successful generative architectures \cite{karras2019style, brock2018large}.  While an exponential linear unit can be used as a differentiable alternative, the Lipschitz constant may still be prohibitively large, requiring small step sizes which prevent efficient exploration of the sample space. We hence take a gradient-free approach to avoid this issue.

\begin{algorithm}[t]
	\caption{Preconditioned Crank--Nicolson}
	\begin{algorithmic}
		\STATE \textbf{Input}: Initial value $Z_0$.
		\FOR{$n \geq 1$}
		\STATE Proposal: $P = \sqrt{1-\beta^2} Z_{n-1} + \beta \xi_n$ where $\xi_n \sim \mathcal{N}(0,C_0)$.
		\STATE Acceptance Probability: $\alpha = \min(1,\exp(\Phi(P) - \Phi(Z_{n-1})))$.
		\STATE Set $Z_n = \begin{cases} P &\mbox{with probability} \  \alpha, \\
		Z_{n-1} &\mbox{otherwise.}
		\end{cases}
		$
		\ENDFOR
		\STATE \textbf{Output}: $\{ Z_n: n \in \mathbb{N} \}$ containing samples from posterior.
	\end{algorithmic}
	\label{pCN}
\end{algorithm}

The choice taken here, to use the preconditioned Crank--Nicolson algorithm \cite{cotter2013mcmc} (pCN),\footnote{The authors are currently investigating a theoretically rigorous variant of the preconditioned Crank--Nicolson algorithm using gradient information by combining explicitly stabilised integration schemes \cite{SKROCK2020} with a control of the Lipschitz constant of the network \cite{ryu2019plug, miyato2018spectral}.} is made possible by choosing the latent distribution $p(z)$ to be Gaussian.  Given a Hilbert space $\mathcal{Z}$, pCN generates a Markov Chain $(Z_n)_{n \in \mathbb{N}}$ with invariant measure $\nu$ of the form:
\[
\nu(E) = \frac{1}{K} \int_E \exp ( - \Phi(z) )\, \nu_0(\diff x),
\]
for each $E \in \mathcal{Z}$, where $\nu_0= \mathcal{N}(0, C_0)$, $K$ is a normalising constant and $\Phi:\mathcal{Z} \to \mathbb{R}$.

The key idea is that, when targeting a measure which is a reweighting of a Gaussian, the performance of standard Metropolis Hastings algorithms can be greatly improved using an Ornstein--Uhlenbeck proposal $Z' = \sqrt{1-\beta^2} Z_{n-1} + \beta \xi_n$ (rather than the symmetric proposal of Random Walk Metropolis (RWM), for example).  To satisfy detailed balance it is necessary to accept with probability $\alpha = \min(1,\exp(\Phi(Z') - \Phi(Z_{n-1})))$.  For RWM and conventional gradient-based proposals,  the acceptance probability tends to zero as dimensionality increases.  In contrast, the acceptance probability and convergence properties of pCN are robust to increasing dimensionality, in the sense that the method is well posed in function spaces \cite{hairer2014spectral}.

\begin{algorithm}[t]
	\caption{Parallel Tempered Preconditioned Crank--Nicolson}
	\begin{algorithmic}
		\STATE \textbf{Input}: Initial values $Z_0^{(i)}$ for  $i = 1,...,k$.
		\STATE \textbf{Input}: Swap proposal frequency $F \in \mathbb{N}$.
		\STATE \textbf{Input}: chain temperatures $0\leq T_0 < T_1 < ... < T_k = 1$.
		\STATE j = 0
		\FOR{$n \geq 1$}
		    \FOR{$i = 1,...,k$}
    		\STATE Proposal: $P^{(i)} = \sqrt{1-\beta^2} Z_{n-1}^{(i)} + \beta \xi_n$ where $\xi_n \sim \mathcal{N}(0,C_0)$.
    		\STATE Acceptance Probability: $\alpha_0 = \min (1,\exp \{T_i [\Phi(P^{(i)}) - \Phi(Z_{n-1}^{(i)})]\} )$
    		\STATE Set $Z_n^{(i)} = \begin{cases} P^{(i)} &\mbox{with probability} \  \alpha_0, \\
    		Z_{n-1} &\mbox{otherwise.}
    		\end{cases}
    		$
    		\ENDFOR
    		\IF{$n \mod F = 0$}{
    		    \STATE Propose temperature swap between chains $j$ and $(j+1)$. (Algorithm \ref{algo:swaps})
    		    \STATE $j = j + 1 \mod k$
    		}
    		\ENDIF
    	\ENDFOR
    	\STATE \textbf{Output}: Chains $\{Z_n^{(i)} : i=1,...,k, n \in \mathbb{N}\}$.
    	\STATE Chain $\{Z_n^{(k)}: n \in \mathbb{N}\}$ contains samples from posterior.
	\end{algorithmic}
	\label{ParallelpCN}
\end{algorithm}
\begin{algorithm}[h!]
	\caption{Temperature Swaps}
	\begin{algorithmic}
		\STATE \textbf{Input}: States $Z^{(i)}$ and $Z^{(j)}$ of chains at temperatures $T_i$ and $T_j$ respectively.
		\STATE Swap proposal: $S^{(i)} = Z^{(j)}, \ S^{(j)} = Z^{(i)}$.
		\STATE Acceptance Probability: $\alpha_1 = \min (1,\exp \{ \Phi(S^{(i)}) - \Phi(S^{(j)})(T_i - T_{j}) \})$.
		\STATE Set $(Z^{(i)}, Z^{(j)}) = \begin{cases} (S^{(i)}, S^{(j)}) &\mbox{with probability} \  \min(1, \alpha_1), \\
			(Z^{(i)}, Z^{(j)}) &\mbox{otherwise.}
		\end{cases}
		$
		\STATE \textbf{Output}: New states $Z^{(i)}$ and $Z^{(j)}$  for chains at temperatures $T_i$ and $T_j$.
	\end{algorithmic}
	\label{algo:swaps}
\end{algorithm}
\begin{figure}[h!]
	\centering
	\begin{subfigure}{.4\textwidth}
		\centering
		\includegraphics[width=.9\textwidth]{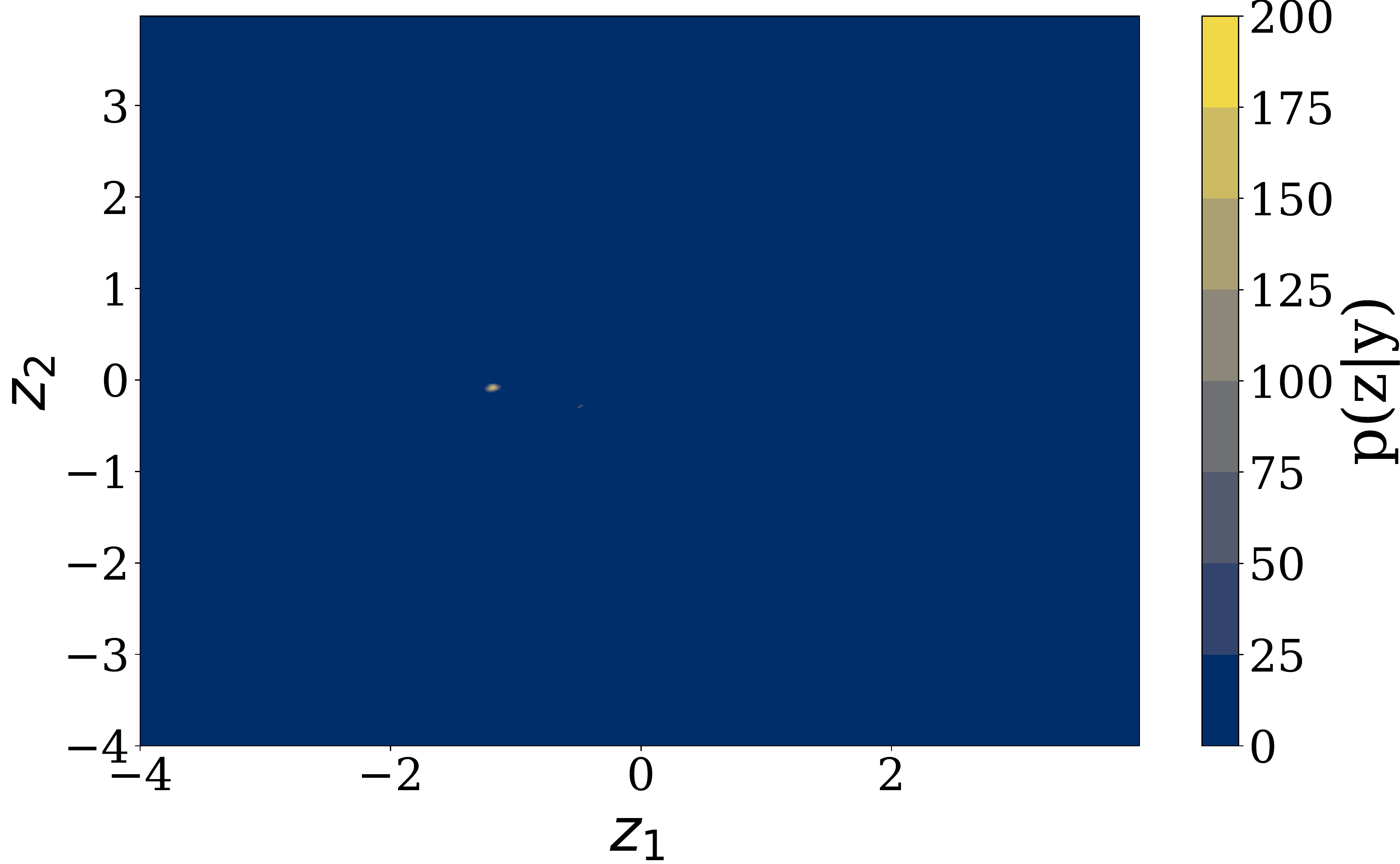}
		\caption{$p(z|y)$}
		\label{subfig:concentratedposterior}
	\end{subfigure}\\
	\begin{subfigure}{.4\textwidth}
		\centering
		\includegraphics[width=.9\textwidth]{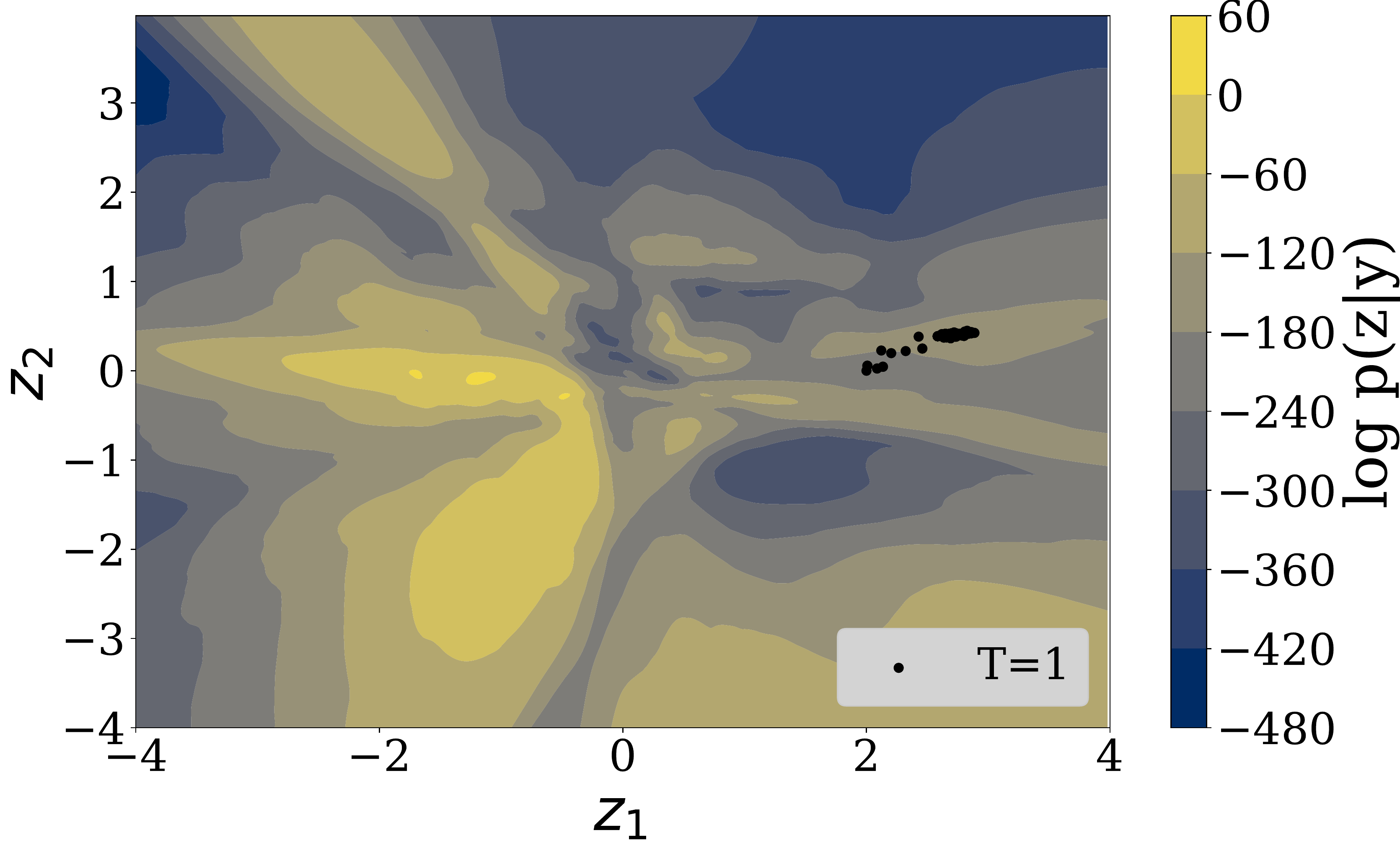}
		\includegraphics[width=.9\textwidth]{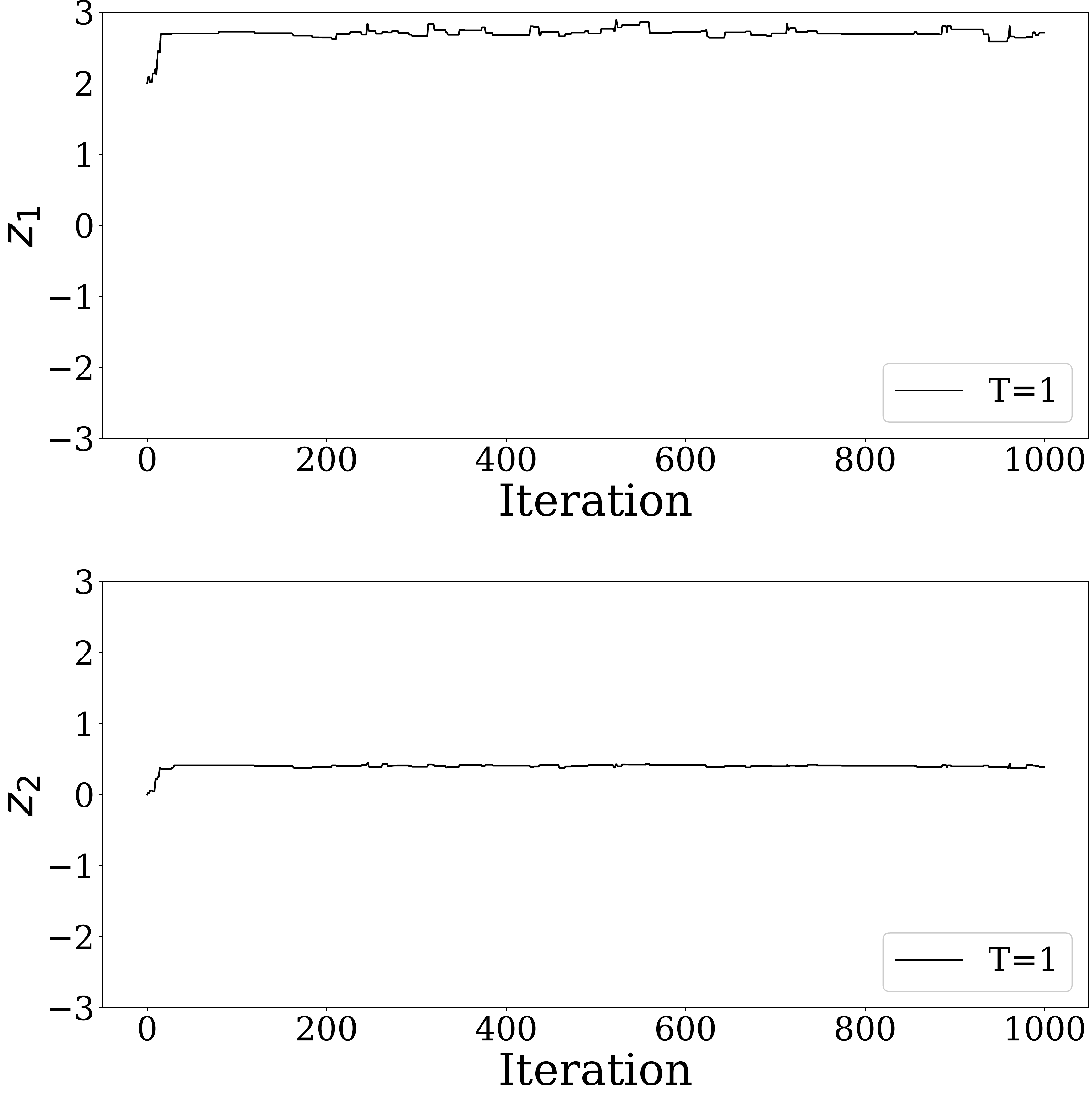}
		\label{subfig:stuck}
		\caption{Single Chain}
	\end{subfigure}
	\begin{subfigure}{.4\textwidth}
		\includegraphics[width=.9\textwidth]{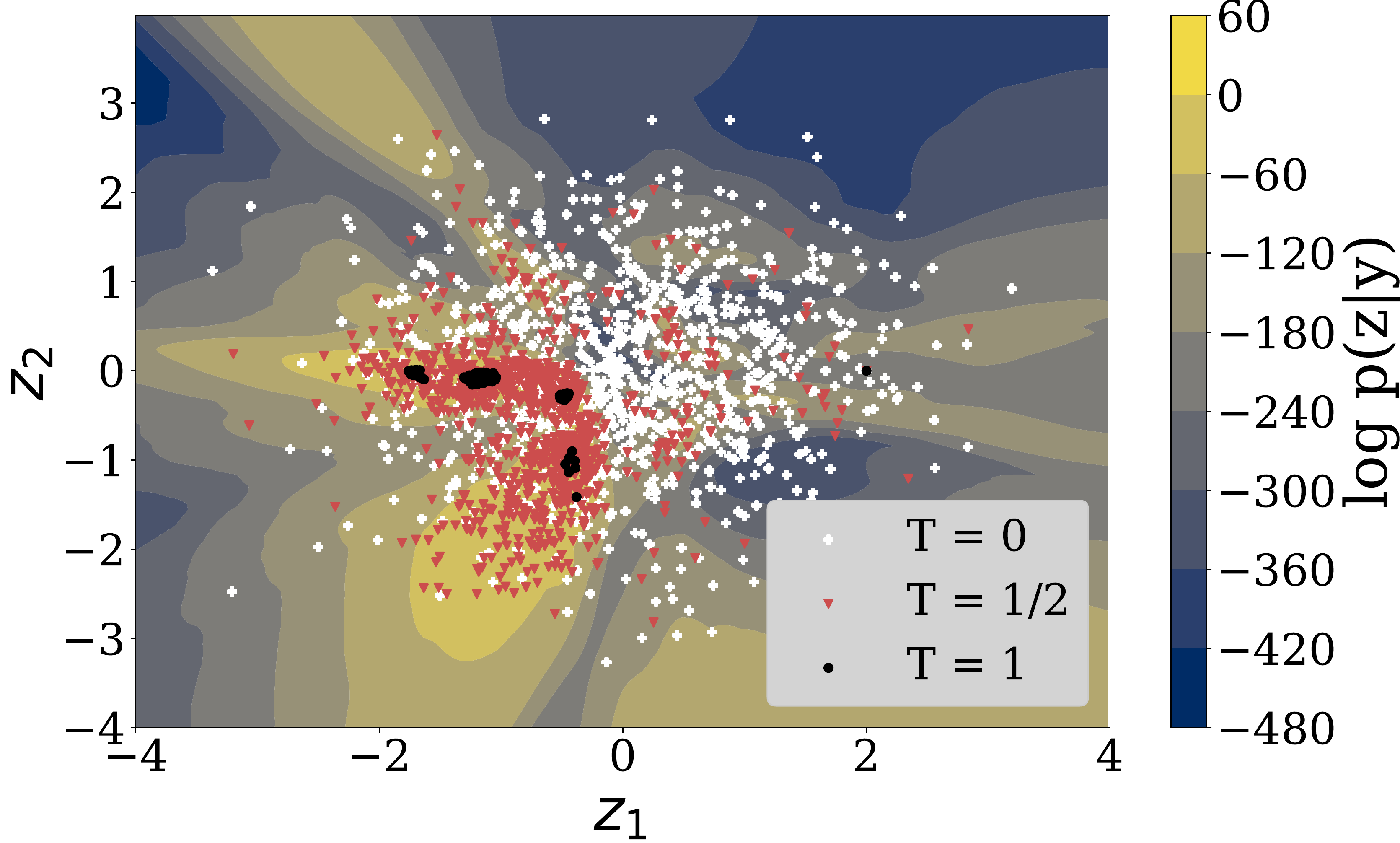}
		\includegraphics[width=.9\textwidth]{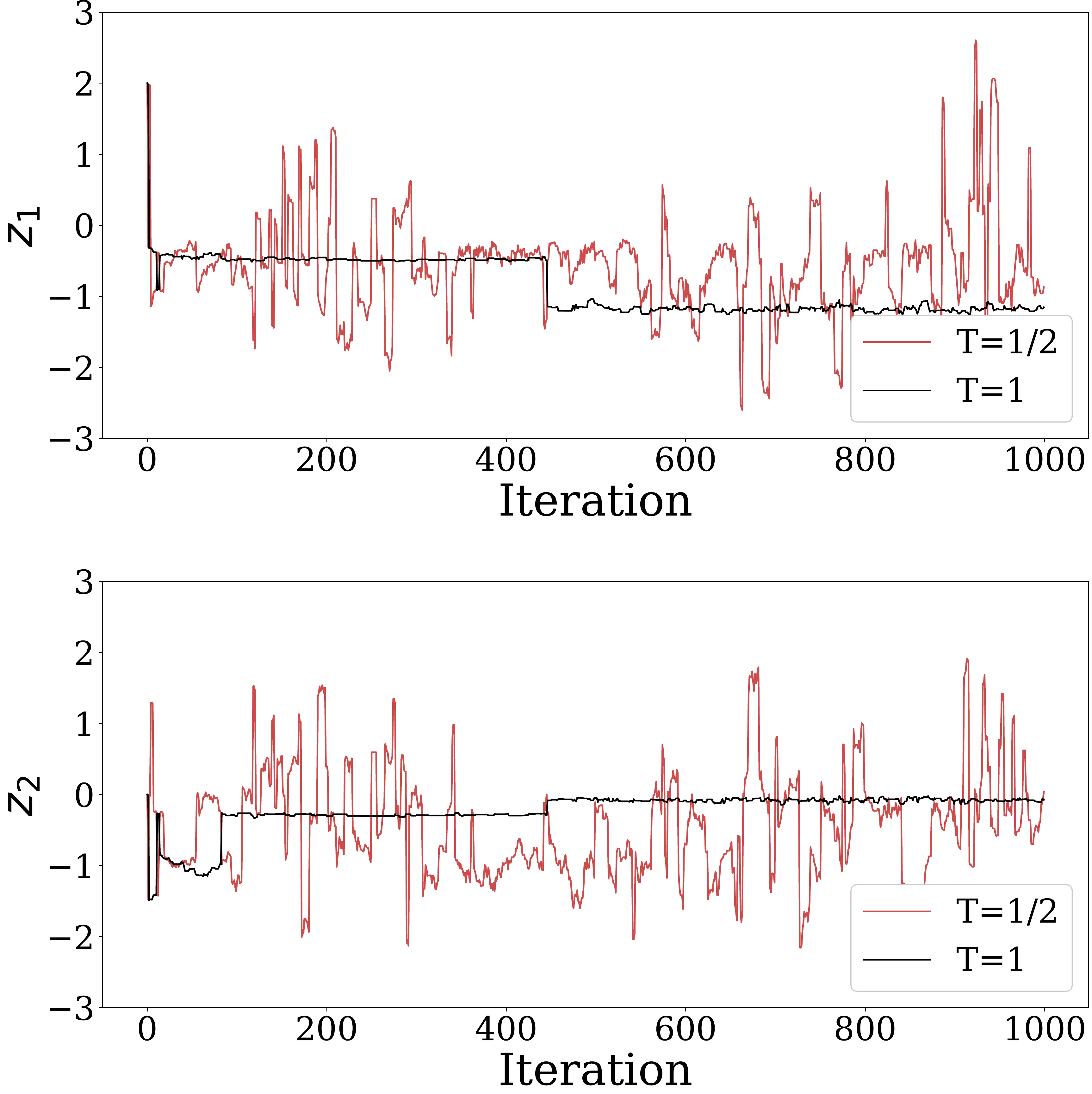}
		\label{subfig:unstuck}
		\caption{Parallel Tempering}
	\end{subfigure}
	\caption{Parallel Tempering: (a) The posterior $p(z|y)$ is highly concentrated, with the near-entirety of the mass located around $(-1.2,0)$.  It is also highly multi-modal, as can be seen when plotted on a log-scale (below).  (b) A single chain initialised in a low probability region (in this case at the point $(2,0)$) can become stuck.  The chain does not reach the region of highest probability around $(-1.2,0)$ after $10^3$ iterations.  (c) Using parallel tempering, from the same initial state, the chains with more energy can escape local maxima.  By swapping states, the chain sampling the true posterior ($T=1$) can reach the region of high probability.}
	\label{fig:paralleltempering}
\end{figure}

Since we expect the nonlinearity of the learned transformation $\mu_\theta$ to result in a multimodal distribution we combine the pCN algorithm with a parallel tempering scheme \cite{geyer1991markov}, running multiple chains in parallel sampling potentials at different `temperatures' $T_i \in [0,1]$:
\[
p(z|y, T_i) = p(y|x_i)^{T_i} p(x).
\]
 Swaps of the states between chains at different temperatures are proposed, and accepted according to a Metropolis rule, enabling transitions between isolated local maxima separated by areas of low probability.  To illustrate this, we performed an experiment using pCN to sample from a highly multimodal 2D posterior, depicted in Figure \ref{fig:paralleltempering}.  After $10^3$ iterations, a single chain initialised in an area of low probability (at the point $z=(2,0)$) is trapped in a local maximum, and has not reached the area of high probability.  In contrast, using 3 chains in a parallel tempering scheme with temperatures $T_1=0$, $T_2=1/2$ and $T_3=1$, enables the chain sampling the true posterior ($T=1$) to escape local maxima, and explore the area of high probability (around $z=(-1.2,0)$).  Examining the trace plots, the chain at $T=1$ appears to exhibit metastability within the first 50 iterations when sampling around $z=(-1/2,-1/4)$, but is able to escape the local maximum and reach another mode of the distribution.

\subsubsection{Ergodicity of Markov Chain}
By construction, the Markov chain generated by the pCN method with parallel tempering admits $p(z|y)$ as invariant density. However, to guarantee convergence to $p(z|y)$ and ensure that the samples generated can be used for Monte Carlo estimates of expectation, we also need to check that the Markov chain is ergodic. For this to hold, it suffices to show that the following conditions are satisfied \cite[Assumptions 6.1 and Theorem 6.2]{cotter2013mcmc}.
\begin{conditions}
The function $\Phi:\mathcal{Z} \to \mathbb{R}$ satisfies the following:
\begin{enumerate}
    \item There exists $p>0$, $K>0$ such that for all $z \in \mathcal{Z}$,
    \[
    0 \leq \Phi(z;y) \leq K ( 1 + \norm{z}^p).
    \]
    \item For every $r > 0$ there is $K(r)>0$ such that for all $z$, $z' \in \mathcal{Z}$ with\\ $\max\{\norm{z}, \norm{z'} \} < r$,
    \[
    \abs{\Phi(z) - \Phi(z')} \leq K(r) \norm{z-z'}.
    \]
\end{enumerate}
\end{conditions}
\begin{remark}
These conditions hold for the Gaussian noise model \ref{eq:gaussianpost}, where $\Phi$ is given by the potential of the likelihood
\[
\Phi(z;y) = -\frac{(y - A \mu_\theta(z))^\dagger (y - A \mu_\theta(z))}{2 \sigma_y^2},
\]
so long as $\mu_\theta$ is Lipschitz continuous, and $\mu_\theta(0)$ is finite.
\begin{proof}
Assume $\mu_\theta$ is $L$-Lipschitz.  Then,
\begin{align*}
    \Phi(z;y) &= \frac{\norm{y-A \mu_\theta(z)}^2}{2\sigma^2} \leq \frac{1}{\sigma^2} ( \norm{y}^2 +  \norm{A \mu_\theta(z)}^2), \\
    &\leq \frac{1}{\sigma^2} \max \{ \norm{y}, L\norm{A} \}^2 (1+\norm{z}^2).
\end{align*}
To prove the second condition holds, we first use the Lipschitz property to show that if $\max \{\norm{z}, \norm{z'} \} < r$, then $\norm{\mu(z) + \mu(z')} \leq \mu_\theta(0) + L \norm{z} + L \norm{z'} < \mu_\theta(0) + L r$.  We then obtain the desired form, assuming $\mu_\theta(0)$ is finite. 
\begin{align*}
    \abs{\Phi(z)-\Phi(z')} &= \frac{1}{2 \sigma^2} \abs{ \norm{y-A \mu_\theta(z)}^2 - \norm{y-A \mu_\theta(z')}^2},\\
    &=  \frac{1}{2 \sigma^2} \abs{ 2 y^\dagger (A(\mu_\theta(z) - \mu_\theta(z')) + \norm{A\mu_\theta(z)}^2 - \norm{A \mu_\theta(z')}^2},\\
    &\leq \frac{1}{2 \sigma^2} \left(2 L \abs{y} \norm{A} \norm{z-z'} + \norm{A}^2 \abs{\mu_\theta(z) - \mu_\theta(z')} \abs{\mu_\theta(z) + \mu_\theta(z')} \right),\\
    & \leq \frac{L \norm{A}}{\sigma^2} \left(\abs{y} + \norm{A} (\abs{\mu_\theta(0)} + Lr) \right)\norm{z-z'}.
\end{align*}
\end{proof}
As stated previously, the conditions of Lipschitz continuity and finiteness are both satisfied when $\mu_\theta$ is a neural network with standard activation functions.
\end{remark}

To conclude, we have shown that a data-driven prior distribution defined by a generative neural network defines a unique posterior distribution in the latent space which is stable with respect to perturbations in the observation and has finite moments, under easily-satisfiable conditions.  We have further established that this posterior distribution may be sampled using a Markov Chain Monte Carlo algorithm which is robust to the multi-modality of the model.

\section{Numerical Experiments} \label{sec:num}

In this section we illustrate the proposed method with an application to the MNIST dataset of $28\times28$ pixel images of handwritten digits \cite{lecun1998mnist}. The pixel values are normalised in the range $[0,1]$  We first discuss the VAE architecture, and in particular the choice of latent dimensionality, before turning to three imaging inverse
problems -- \emph{denoising, inpainting and deconvolution}.  Such problems are traditionally solved by MAP estimation in the variational formulation. Here we demonstrate the benefits of our Bayesian sampling approach, which can be used to obtain point estimates, as well as to visualise uncertainty, and perform  advanced Bayesian analyses beyond the scope of optimisation-based mathematical imaging methodologies.

\subsection{VAE Training}
\subsubsection{Architecture}
The prior takes the form of a Variational Autoencoder with a standard normal marginal latent distribution $p(z) = p_\mathcal{N}(z;0,\mathbb{I}_m)$.  The VAE is made up of two fully connected networks, each with two hidden layers and RELU activation functions. The encoder is a fully connected [784 --- 512 --- 512 --- $2m$] network, where $m$ denotes the latent dimensionality and $784=28\times28$ is the size of the input image.  The output layer has dimension $2m$, returning the mean, $\mu_\phi(x)$, and variance, $\sigma_\phi(x)$, of the encoded distribution $q_\phi(z|x)$. The decoder is also fully connected with the architecture [$m$ -- 512 -- 512 -- 784].  In section \ref{subsec:latentDim} we show how $m$ can be chosen using a method to determine the intrinsic dimensionality of the data.  The networks are trained using Adam \cite{kingma2014adam} with default parameters, with a minibatch size of 128, over 50 epochs.  Training, on a Intel(R) Xeon(R) CPU @ 2.30GHz and Tesla P100 16GB GPU, lasted only 112.2 seconds.
 
\subsubsection{Determining latent dimensionality} \label{subsec:latentDim}

\begin{figure}[t]
	\centering
	\begin{subfigure}{.48\textwidth}
		\centering
		\includegraphics[width=\textwidth]{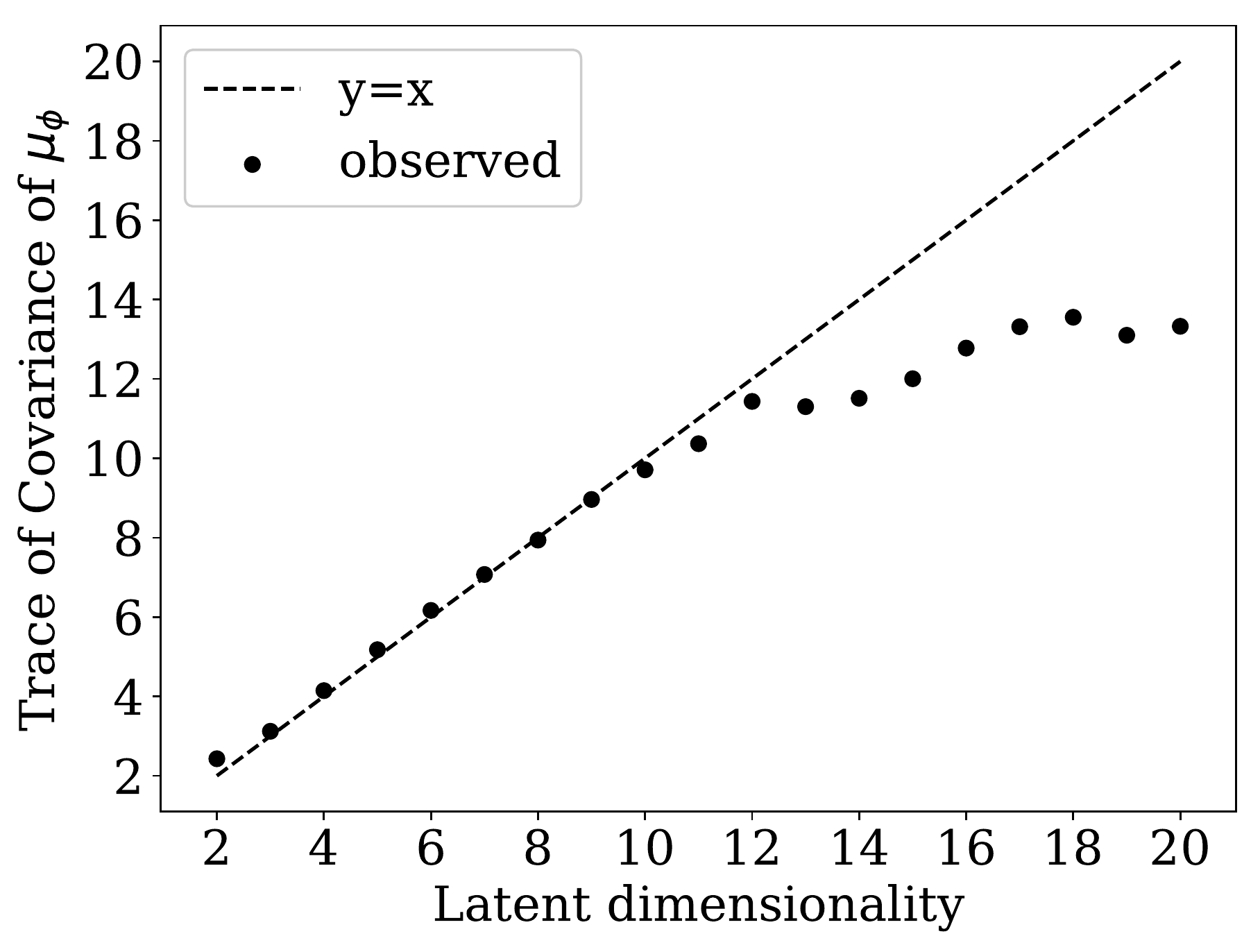}
		\caption{}
		\label{subfig:covariance}
	\end{subfigure}
	\begin{subfigure}{.48\textwidth}
		\centering
		\includegraphics[width=\textwidth]{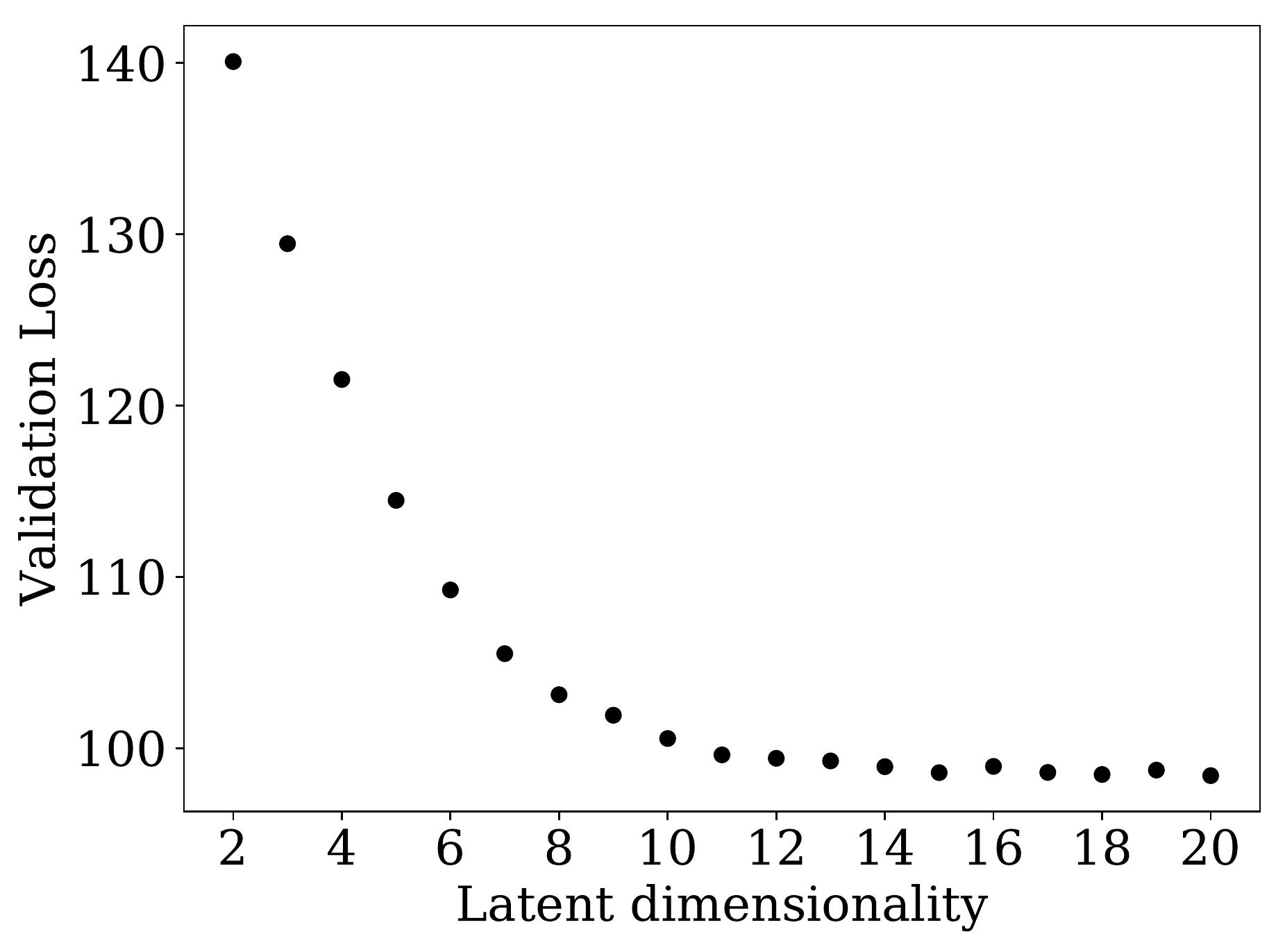}
		\caption{}
		\label{subfig:loss}
	\end{subfigure}	
	\caption{Plots against dimension of (a) trace of sample covariance of $\mu_\phi(x_i)$ across all test images (b) validation loss.  For dimensions greater than 12, the validation loss does not significantly decrease. The  encoded means for dimensionality greater than 12 suggest additional dimensionality is unnecessary, since the means of the encoded distribution in additional dimensions have low variance across all images in the testing set.}
	\label{fig:choosing dim}
\end{figure}

One important choice in the architecture of the variational autoencoder is the dimensionality of the latent space.  To accurately approximate the true distribution, the latent space should have dimensionality at least equal to the intrinsic dimensionality of the data.  On the other hand, overestimating the dimensionality results in less regularisation (and a more computationally challenging sampling problem).  We here present a simple empirical check which can help identify the dimensionality of the latent space.  This relies on the encoder of the VAE, and illustrates one benefit of this approach over alternative generative models.

The dimensionality check proceeds as follows. Applying the encoder to a test set of images $\{x_i\}_{i=1}^n$ returns a corresponding set of encoded means and variances which define the distributions $q_\phi(z|x_i) = p_\mathcal{N}(z;\mu_\phi(x_i), \sigma_\phi^2(x_i) \mathbb{I}_m)$. For a VAE that accurately captures the manifold supporting $\{x_i\}_{i=1}^n$, we expect that the means $\{\mu_\phi(x_i)\}_{i=1}^n$ will carry most of the information and dominate the variances $\{\sigma_\phi^2(x_i)\}_{i=1}^n$, which are related to the uncertainty of the encoding. This can be assessed by calculating the trace of the sample covariance of $\{\mu_\phi(x_i)\}_{i=1}^n$, which for an accurately fitted model should be close to $\textrm{trace}(\mathbb{I}_m)=m$, as that is the target latent covariance for the VAE. If the trace of the sample covariance is less than the latent dimensionality, this indicates redundant dimensionality; i.e. the dimensionality of the latent space exceeds the intrinsic dimensionality of the manifold supporting $\{x_i\}_{i=1}^n$. In the case of MNIST, we observe in Figure \ref{subfig:covariance} this effect for dimensions greater than 12, and therefore choose a VAE with a 12-dimensional latent space. Similarly, Figure \ref{subfig:loss} illustrates that increasing the dimensionality above 12 provides little benefit in validation loss. 

\subsection{Practical Implementation Guidelines for pCN}
To optimise computational efficiency, both the step size in the pCN algorithm and the number of chains in the parallel tempering scheme can be tuned to aim for an acceptance rate, $a$, in the region of 0.2 to 0.5. We use 25\% as default, following the optimal scaling rule derived for the Random Walk Metropolis algorithm \cite{roberts1997weak} and the similar result obtained for parallel tempered MCMC \cite{atchade2011towards}.  We use automatic step size selection according to a version of the Robbins--Monro algorithm \cite{robbins1951stochastic} adapted for the constraint $\beta \in (0,1)$,  as follows:
\begin{align*}
	\delta_n &= \log \frac{\beta_n}{1 - \beta_n},\\
	\delta_{n+1} &= \delta_n + c_n(\alpha(z',z_{n-1})) - a),\\
	\beta_{n+1} &= \frac{1}{1+e^{-\delta_{n+1}}},
\end{align*}
where $\{c_n\}_{n=1}^\infty$ is a sequence of positive real numbers such that $\sum_{n=0}^\infty c_n = \infty$ and $\sum_{n=0}^\infty c_n^2 < \infty$.  For the experiments in this paper, $c_n = c/n$ for constant $c$.

\subsection{Point Estimates} \label{subsec:ptEstimates}

To obtain point estimates based on the posterior \ref{eq:latentpost}, we use the minimum mean square error (MMSE) estimator
\[
\hat{x} = \argmin_{x'} \int \abs{\mu_\theta(z) - x'}^2  p(z|y)\, dz,
\]
which is equal to the mean of the posterior distribution.  This estimator is approximated using posterior samples by the Monte Carlo estimator
\[
x_{MMSE} = \frac{1}{n} \sum_{i=1}^n \mu_\theta(z_i),
\]
here $\{z_i\}_{i=1}^N$ are samples from the latent posterior with density $p(z|y)$.  We note that it would also be possible to define an estimator by $\mu_\theta(\hat{z})$, where $\hat{z} =  \frac{1}{n} \sum_{i=1}^n z_i$, but this would force the estimates to lie in the range of the function $\mu_\theta$. Since the quantity of interest is the image $x$, it is more natural to calculate the mean in the pixel space, rather than the latent space. 

\def\arraystretch{0.4}
\begin{figure}[t]
	\centering
	\begin{minipage}{\textwidth}
		\centering
		\begin{tabular}[t]{l@{\hskip0pt} @{\hskip0pt}c @{\hskip0pt}c @{\hskip0pt}c  c @{\hskip0pt}c @{\hskip0pt}c}
			\setlength\tabcolsep{1 pt}
			True Image &
			&
			\includegraphics[width=1.5cm]{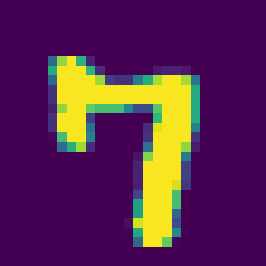}&
			&
			&
			\includegraphics[width=1.5cm]{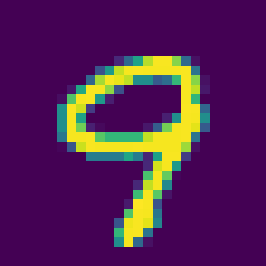}&
			\\
			Observation &
			\includegraphics[width=1.5cm]{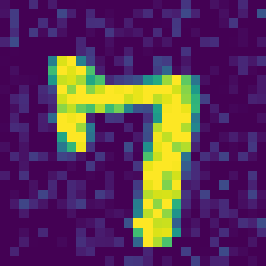}&
			\includegraphics[width=1.5cm]{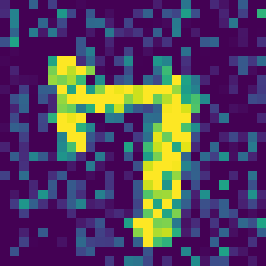}&
			\includegraphics[width=1.5cm]{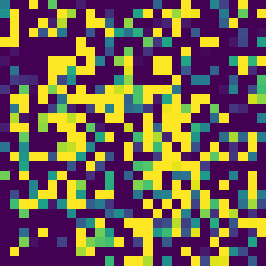}&
			\includegraphics[width=1.5cm]{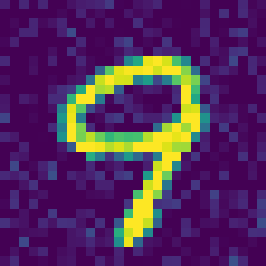}&
			\includegraphics[width=1.5cm]{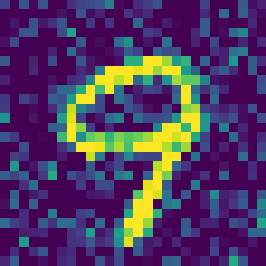}&
			\includegraphics[width=1.5cm]{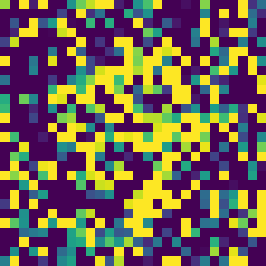}\\
			& \scriptsize $\sigma=0.1$ & \scriptsize $\sigma=0.25$ & \scriptsize $\sigma=1.0$ &  \scriptsize $\sigma=0.1$ & \scriptsize $\sigma=0.25$ & \scriptsize $\sigma=1.0$\\
			DAE \cite{vincent2010stacked} &
			\includegraphics[width=1.5cm]{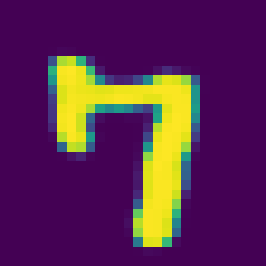}&
			\includegraphics[width=1.5cm]{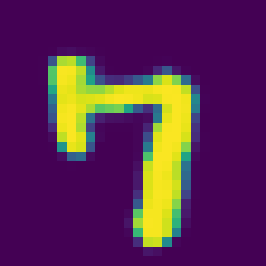}&
			\includegraphics[width=1.5cm]{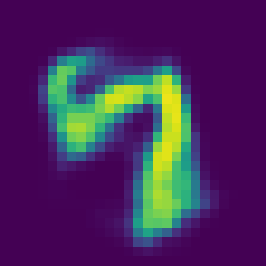}&
			\includegraphics[width=1.5cm]{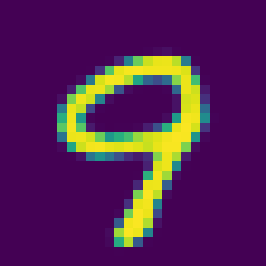}&
			\includegraphics[width=1.5cm]{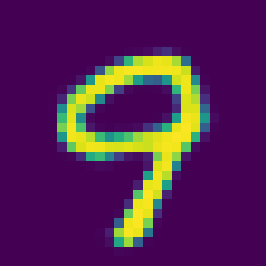}&
			\includegraphics[width=1.5cm]{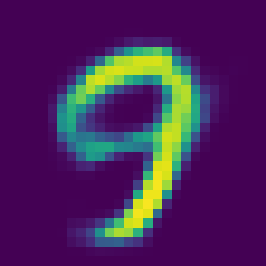}\\
			& \scriptsize 27.81 dB & \scriptsize 25.57 dB & \scriptsize 17.22 dB & \scriptsize 27.52 dB & \scriptsize 23.00 dB & \scriptsize 14.85 dB \\
			JPMAP \cite{gonzalez2019solving} &
			\includegraphics[width=1.5cm]{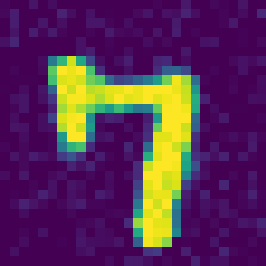}&
			\includegraphics[width=1.5cm]{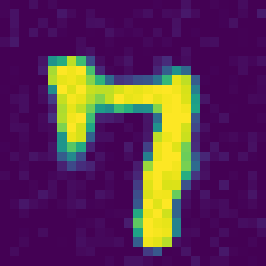}&
			\includegraphics[width=1.5cm]{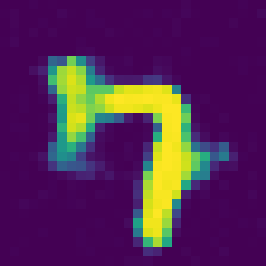}&
			\includegraphics[width=1.5cm]{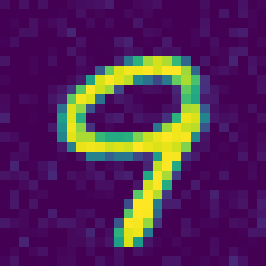}&
			\includegraphics[width=1.5cm]{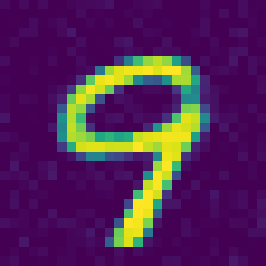}&
			\includegraphics[width=1.5cm]{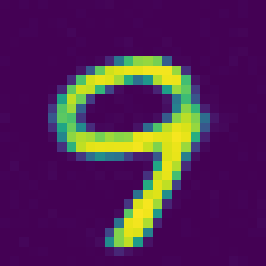}\\
			& \scriptsize 25.06 dB & \scriptsize 24.16 dB & \scriptsize 16.87 dB & \scriptsize 23.69 dB & \scriptsize 21.50 dB & \scriptsize 15.88 dB\\
			MMSE (Ours) &
			\includegraphics[width=1.5cm]{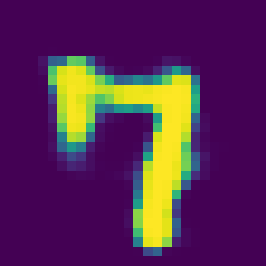}&
			\includegraphics[width=1.5cm]{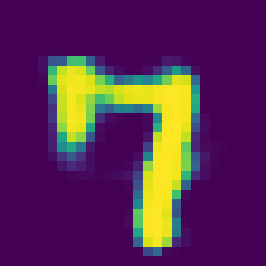}&
			\includegraphics[width=1.5cm]{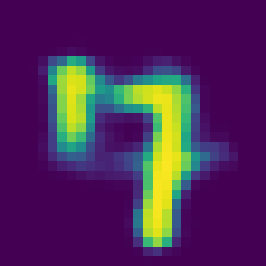}&
			\includegraphics[width=1.5cm]{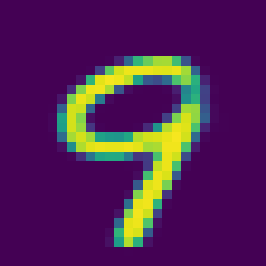}&
			\includegraphics[width=1.5cm]{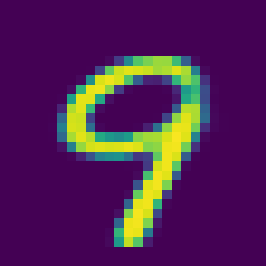}&
			\includegraphics[width=1.5cm]{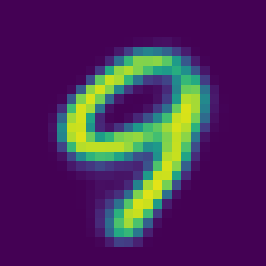}\\
			& \scriptsize 24.99 dB & \scriptsize 23.99 dB & \scriptsize 17.58 dB & \scriptsize 21.94 dB & \scriptsize 21.85 dB & \scriptsize 18.74 dB	
		\end{tabular}
	\end{minipage}\hfill
	\caption{Denoising: From left to right, s.d. of noise, $\sigma=0.1,\  0.25,\  1.0$. PSNR reported below reconstructions.}
	\label{fig:denoising}
\end{figure}

It should be noted, however, that under this model we cannot obtain a maximum a posteriori probability (MAP) estimate in the image space, since the image posterior does not admit a density w.r.t. the appropriate Lebesgue measure. A MAP estimate of the latent variable could be obtained, but this highlights a well known drawback of the MAP estimator -- it is not invariant to reparametrisation \cite{druilhet2007invariant}.  If $z'$ is the mode of the latent posterior, there is no corresponding Bayesian interpretation for the generated image  $\mu_\theta(z')$. We hence only report results on the MMSE estimator, which has a strong decision-theoretic justification as a Bayesian estimator \cite{robert2007bayesian}.  

In the following, we benchmark the results obtained by the point estimate defined by the MMSE estimator against two different approaches.  The first is a plug-and-play (PnP) ADMM algorithm \cite{venkatakrishnan2013plug}, using as denoiser a denoising autoencoder \cite{vincent2010stacked} trained on MNIST with Gaussian noise of standard deviation between $10^{-2}$ and 1.  The second is the more closely-related joint posterior maximum a posteriori probability (JPMAP) method \cite{gonzalez2019solving}, based on a generative prior using the same VAE.  Both of these methods are qualitatively different from our own.  The JPMAP estimator \cite{gonzalez2019solving} finds the MAP of the joint posterior distribution $p(z,x | y)$ defined by a stochastic VAE, while the plug-and-play ADMM approach \cite{venkatakrishnan2013plug} can also be seen as a generalisation of a variational problem, where the proximal of the ADMM algorithm (which can be viewed as a specific denoising operator) is replaced by a denoising operator defined by a neural network.

Figures \ref{fig:denoising}, \ref{fig:inpainting} and \ref{fig:deblurring} illustrate the performance of each method on a variety of images and inverse problems, exhibiting the different behaviours of the algorithms.  Alongside these results, we report the mean PSNR obtained by our method over the entire test dataset of $10^4$ images (Table \ref{PSNRtable}).  Since hyperparameter tuning was required for both alternative methods to achieve optimal results, it was not possible to perform this exhaustive test on the alternative methods.  In each experiment, we use 10 parallel chains, with an initial step size of 0.1, and initialize at zero.  The initial $2\times 10^4$ samples are discarded as burn-in, to ensure the Markov Chain has reached stationarity.

\paragraph{Denoising}

\def\arraystretch{0.4}
\begin{figure}[t]
	\centering
	\begin{minipage}{\textwidth}
		\centering
		\begin{tabular}{l@{\hskip0pt} @{\hskip0pt}c @{\hskip0pt}c @{\hskip0pt}c  c @{\hskip0pt}c @{\hskip0pt}c}
			\setlength\tabcolsep{1 pt}
			True Image &
			&
			\includegraphics[width=1.5cm]{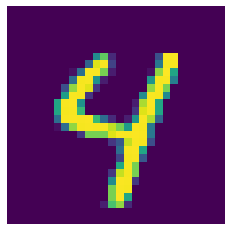}&
			&
			&
			\includegraphics[width=1.5cm]{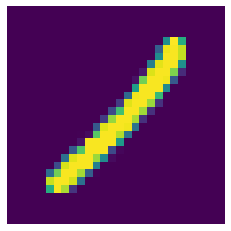}&
			\\
			Observation &
			\includegraphics[width=1.5cm]{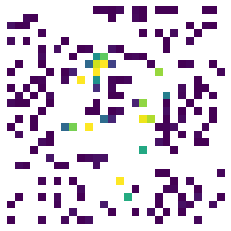}&
			\includegraphics[width=1.5cm]{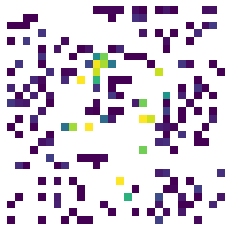}&
			\includegraphics[width=1.5cm]{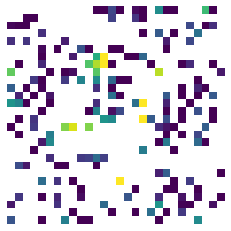}&
			\includegraphics[width=1.5cm]{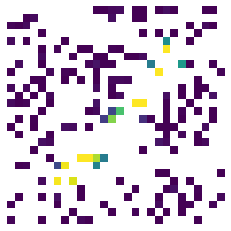}&
			\includegraphics[width=1.5cm]{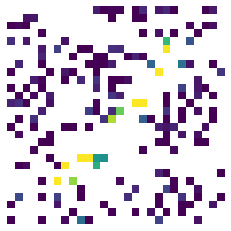}&
			\includegraphics[width=1.5cm]{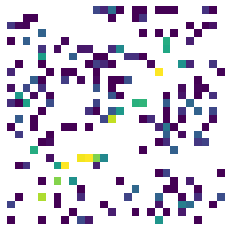}\\
			& \scriptsize $\sigma=0.025$ & \scriptsize $\sigma=0.1$ & \scriptsize $\sigma=0.25$ &  \scriptsize $\sigma=0.025$ & \scriptsize $\sigma=0.1$ & \scriptsize $\sigma=0.25$\\
			PnP ADMM \cite{venkatakrishnan2013plug} &
			\includegraphics[width=1.5cm]{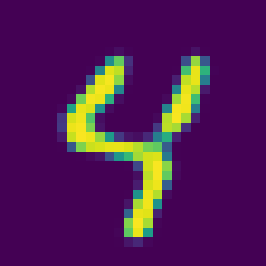}&
			\includegraphics[width=1.5cm]{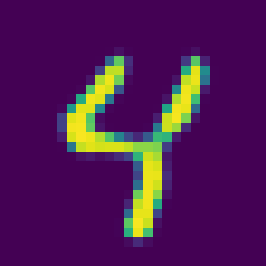}&
			\includegraphics[width=1.5cm]{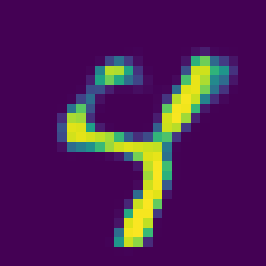}&
			\includegraphics[width=1.5cm]{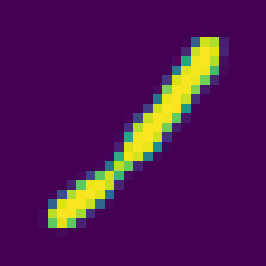}&
			\includegraphics[width=1.5cm]{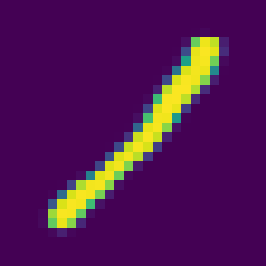}&
			\includegraphics[width=1.5cm]{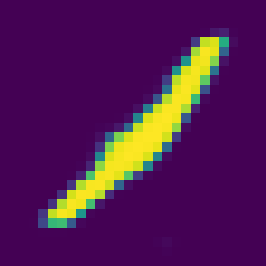}\\
			& \scriptsize 19.90 dB & \scriptsize 18.68 dB & \scriptsize 17.58 dB & \scriptsize 22.19 dB & \scriptsize 22.54 dB & \scriptsize 20.80 dB \\
			JPMAP \cite{gonzalez2019solving} &
			\includegraphics[width=1.5cm]{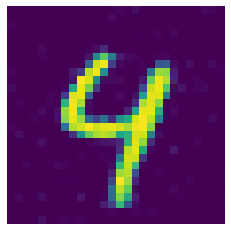}&
			\includegraphics[width=1.5cm]{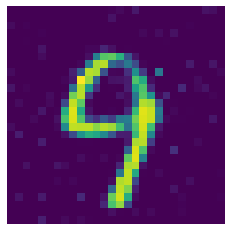}&
			\includegraphics[width=1.5cm]{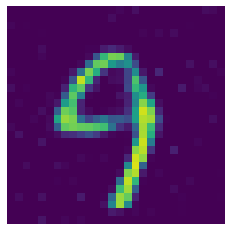}&
			\includegraphics[width=1.5cm]{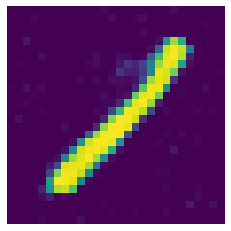}&
			\includegraphics[width=1.5cm]{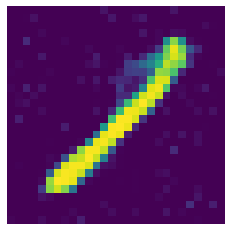}&
			\includegraphics[width=1.5cm]{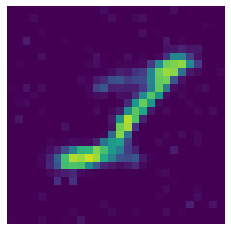}\\
			& \scriptsize 19.01 dB & \scriptsize 14.85 dB & \scriptsize 14.79 dB & \scriptsize 24.26 dB & \scriptsize 22.20 dB & \scriptsize 14.22 dB\\
			MMSE (Ours) &
			\includegraphics[width=1.5cm]{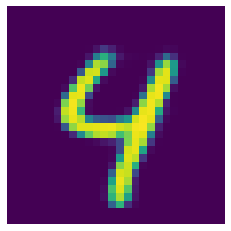}&
			\includegraphics[width=1.5cm]{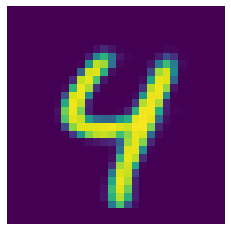}&
			\includegraphics[width=1.5cm]{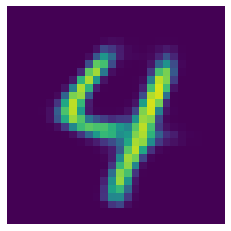}&
			\includegraphics[width=1.5cm]{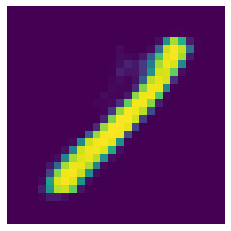}&
			\includegraphics[width=1.5cm]{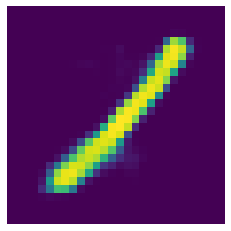}&
			\includegraphics[width=1.5cm]{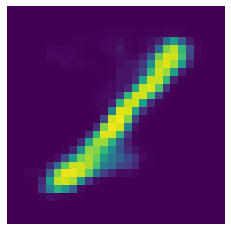}\\
			& \scriptsize 21.36 dB & \scriptsize 20.47 dB & \scriptsize 18.54 dB & \scriptsize 24.88 dB & \scriptsize 24.39 dB & \scriptsize 20.08 dB	
		\end{tabular}
	\end{minipage}\hfill
	\caption{Inpainting with 25\% known pixels. From left to right, s.d. of noise, $\sigma=0.01,\  0.025,\  0.1$. PSNR reported below reconstructions.}
	\label{fig:inpainting}
\end{figure}

We first consider the problem of Gaussian denoising
\[
y = x + w \quad \text{where} \quad w \sim N(0, \sigma^2 \mathbb{I}_p).
\]

For denoising, the MMSE estimator is calculated from $5 \times 10^4$ samples, (Figure \ref{fig:denoising}).  Since the task is denoising, for which the autoencoder was trained, rather than comparing to the PnP ADMM method, we display the results for the denoising autoencoder itself.  For the easiest case, of removing low-variance noise, the methods all perform well, and obtain accurate reconstructions. As expected, having been trained on the problem, the end-to-end denoiser attains the best results. However, for more aggressive noise levels ($\sigma=1.0$, the upper limit of the training range), the performance of the end-to-end denoiser deteriorates. The proposed MMSE estimator, on the other hand, is significantly more robust to the more challenging levels of noise, highlighting the strength of the Bayesian framework in problems with higher uncertainty. The results obtained by JPMAP are more robust to high noise than the end-to-end network, but for low noise they retain some of the noise in the observation, a result of the use of the stochastic decoder with an unstructured covariance matrix.

\paragraph{Inpainting}
Inpainting is defined by the forward model
\[
y = M x + w \quad \text{where} \quad w \sim N(0, \sigma^2 \mathbb{I}_p),
\]
where $M$ is a diagonal masking matrix, taking value 1 where pixels are observed, and 0 for missing pixels.  For our experiments, we observe 25\% of pixels at random.  The MMSE estimator is calculated from $10^5$ samples (Figure \ref{fig:inpainting}).  We observe that the proposed Bayesian method delivers accurate estimation results in all cases, and often significantly outperforms the alternative approaches. JPMAP exhibits a very competitive performance for low-noise problems, whereas PnP ADMM is more competitive in problems with higher noise levels.

\def\arraystretch{0.4}
\begin{figure}[t]
	\centering
	\begin{minipage}{\textwidth}
		\centering
		\begin{tabular}{l@{\hskip0pt} @{\hskip0pt}c @{\hskip0pt}c @{\hskip0pt}c  c @{\hskip0pt}c @{\hskip0pt}c}
			\setlength\tabcolsep{1 pt}
			True Image &
			&
			\includegraphics[width=1.5cm]{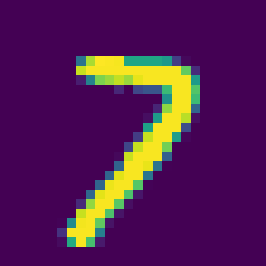}&
			&
			&
			\includegraphics[width=1.5cm]{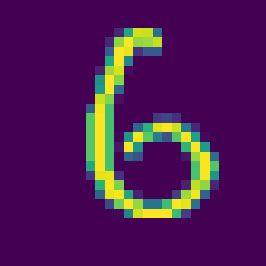}&
			\\
			Observation &
			\includegraphics[width=1.5cm]{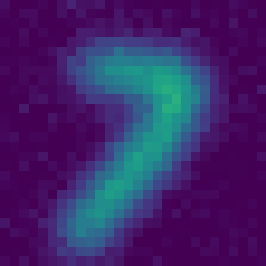}&
			\includegraphics[width=1.5cm]{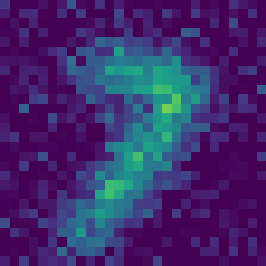}&
			\includegraphics[width=1.5cm]{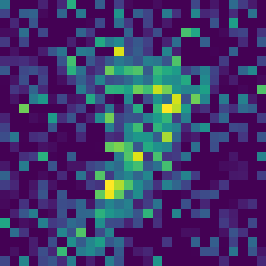}&
			\includegraphics[width=1.5cm]{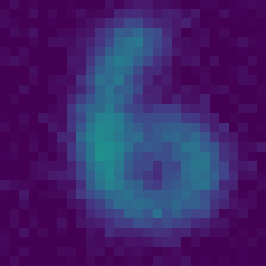}&
			\includegraphics[width=1.5cm]{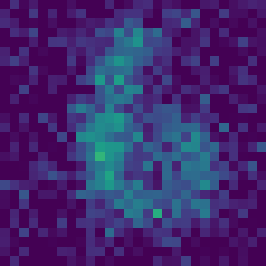}&
			\includegraphics[width=1.5cm]{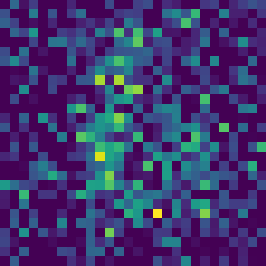}\\
			& \scriptsize $\sigma=0.025$ & \scriptsize $\sigma=0.1$ & \scriptsize $\sigma=0.25$ &  \scriptsize $\sigma=0.025$ & \scriptsize $\sigma=0.1$ & \scriptsize $\sigma=0.25$\\
			PnP ADMM \cite{venkatakrishnan2013plug} &
			\includegraphics[width=1.5cm]{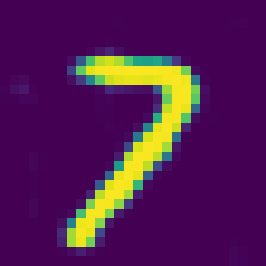}&
			\includegraphics[width=1.5cm]{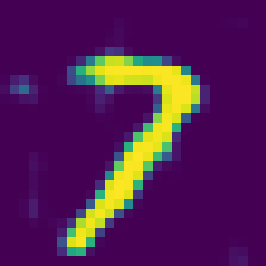}&
			\includegraphics[width=1.5cm]{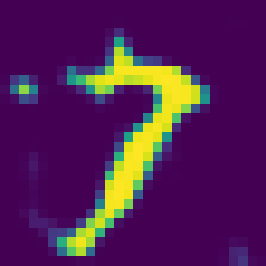}&
			\includegraphics[width=1.5cm]{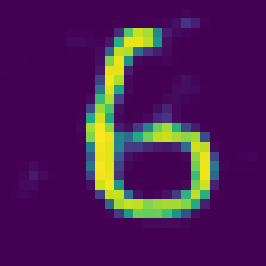}&
			\includegraphics[width=1.5cm]{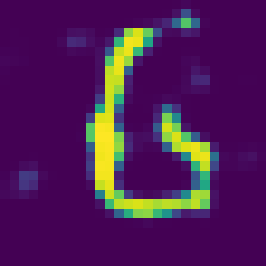}&
			\includegraphics[width=1.5cm]{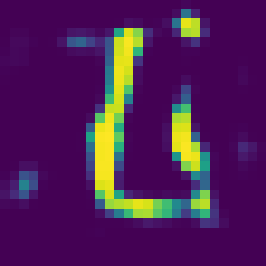}\\
			& \scriptsize 28.58 dB & \scriptsize 20.73 dB & \scriptsize 14.47 dB & \scriptsize 21.28 dB & \scriptsize 16.08 dB & \scriptsize 12.97 dB \\
			JPMAP \cite{gonzalez2019solving} &
			\includegraphics[width=1.5cm]{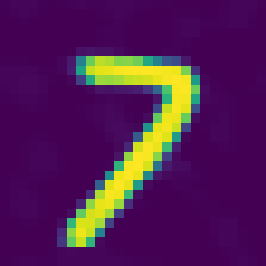}&
			\includegraphics[width=1.5cm]{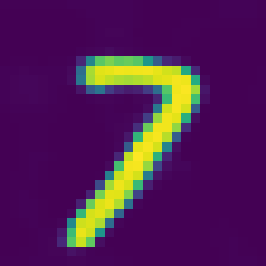}&
			\includegraphics[width=1.5cm]{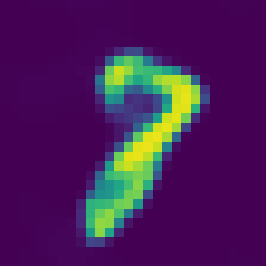}&
			\includegraphics[width=1.5cm]{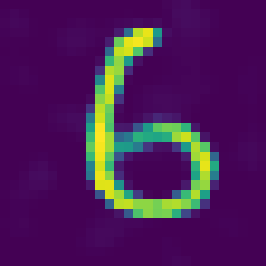}&
			\includegraphics[width=1.5cm]{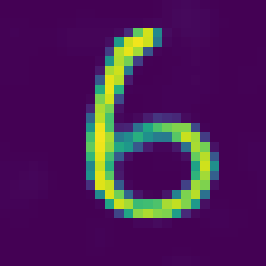}&
			\includegraphics[width=1.5cm]{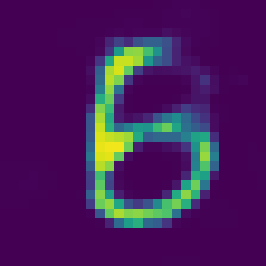}\\
			& \scriptsize 27.81 dB & \scriptsize 23.17 dB & \scriptsize 15.20 dB & \scriptsize 22.82 dB & \scriptsize 21.46 dB & \scriptsize 14.92 dB\\
			MMSE (Ours) &
			\includegraphics[width=1.5cm]{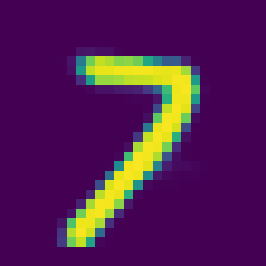}&
			\includegraphics[width=1.5cm]{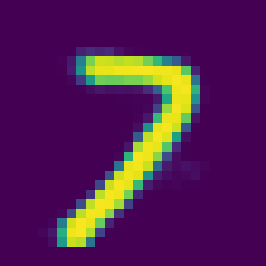}&
			\includegraphics[width=1.5cm]{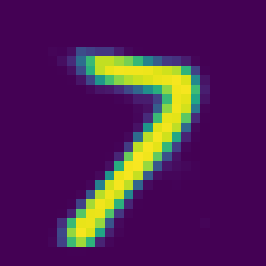}&
			\includegraphics[width=1.5cm]{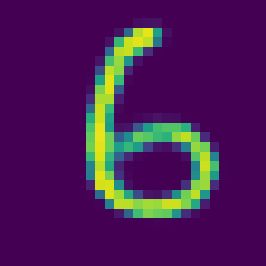}&
			\includegraphics[width=1.5cm]{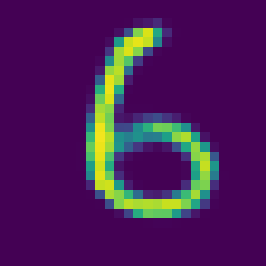}&
			\includegraphics[width=1.5cm]{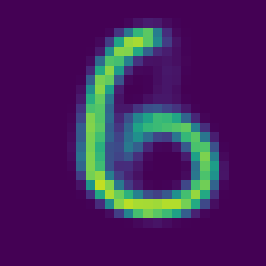}\\
			& \scriptsize 28.07 dB & \scriptsize 25.08 dB & \scriptsize 24.41 dB & \scriptsize 22.64 dB & \scriptsize 21.21 dB & \scriptsize 17.57 dB	
		\end{tabular}
	\end{minipage}\hfill
	\caption{Deblurring $6 \times 6$ Gaussian blur with standard deviation of 2. From left to right, s.d. of noise, $\sigma=0.01,\  0.025,\  0.1,\  0.25$. PSNR reported below reconstructions.}
	\label{fig:deblurring}
\end{figure}

\paragraph{Deblurring}
The third problem considered is deblurring, defined by the forward model
\[
y_i = G x + w \quad \text{where} \quad w \sim N(0, \sigma^2 \mathbb{I}_p),
\]
where $G$ denotes convolution with a $6\times6$ Gaussian blur kernel with a bandwidth of $2$ pixels.  The ill-posedness of deconvolution requires a smaller step size, and thus more samples to fully explore the posterior.  In this case the MMSE estimator is calculated from $2 \times 10^5$ samples (Figure \ref{fig:deblurring}).  Once again, we obtain consistently accurate results, and outperform the optimisation-based methods in the most challenging problems with high noise levels.

\begin{table}[h!]
	\centering
	\def\arraystretch{1.0}
	\begin{tabular}{ | c ||  c | c |  c | c | }
		\hline
		$\sigma$ & 0.1 & 0.25 & 1.0 & 2.5\\ \hline
		Denoising & $22.53 \pm 3.35$ & $21.58 \pm 2.79$ & $17.17 \pm 2.28$ & $13.29 \pm 1.51$ \\ \hline
	\end{tabular}\\
	\hfill\\
	\hfill\\
	\begin{tabular}{ | c ||  c | c |  c | c |  }
		\hline
		$\sigma$ & 0.01 & 0.025 & 0.1 & 0.25 \\ \hline
		Inpainting & $19.90 \pm 2.96$ & $20.00 \pm 2.90$ & $19.96 \pm 2.83$ & $19.08 \pm 2.23$ \\ \hline
		Deblurring & $21.48 \pm 3.36$ & $21.40 \pm 3.20$ & $20.17 \pm  2.49$ & $17.64 \pm 1.99$ \\ \hline
	\end{tabular}
	\caption{PSNR of MMSE estimator for $10^4$ test images for varying noise levels.}
	\label{PSNRtable}
\end{table}

\subsection{Visualising Uncertainty}

A major benefit of the Bayesian framework is its ability to capture the uncertainty in the delivered solutions.  To visualise this uncertainty, we use the generated samples in order to perform a principal component analysis (PCA) of the posterior covariance matrix. This then allows identifying the leading eigenvectors of the covariance and visualising the uncertainty as projected on two-dimensional subspaces.  In detail, projecting from the latent space to these two-dimensional subspaces, images generated by the sampling algorithm can then be plotted according to their location in the 2D space defined by the principal components to create a map illustrating the uncertainty in the posterior samples. Figure \ref{pca0} below depicts the uncertainty associated with the two leading eigenvectors of the posterior covariance for an image deblurring problem with $\sigma = 0.5$.

\begin{figure}[h!]
	\centering
		\begin{subfigure}{.2\textwidth}
			\centering
			\includegraphics[width=\textwidth]{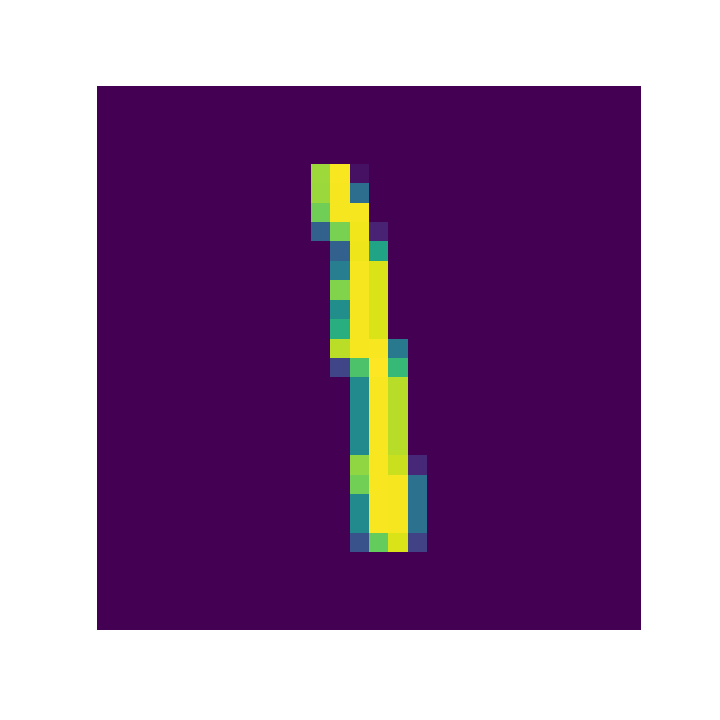}
			\caption*{True}
		\end{subfigure}
		\begin{subfigure}{.2\textwidth}
			\centering
			\includegraphics[width=\textwidth]{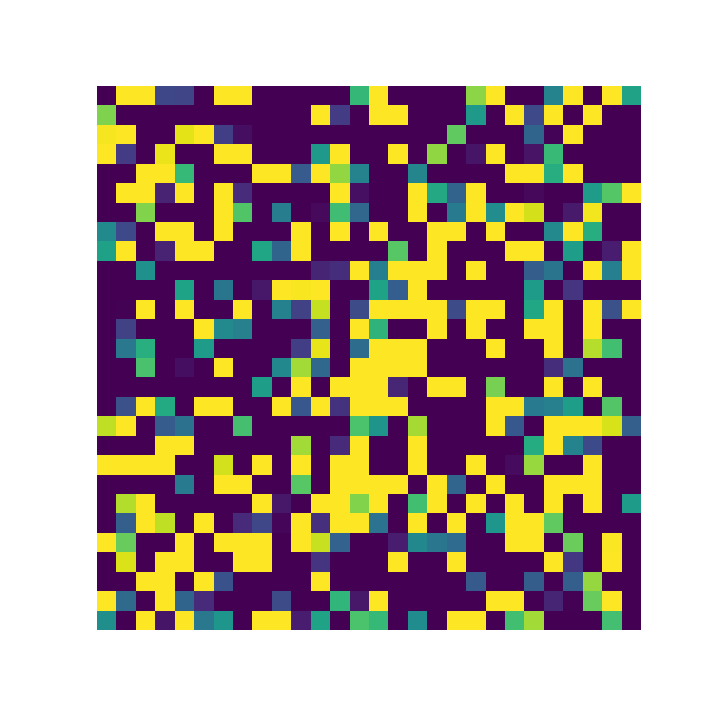}
			\caption*{Observation}
		\end{subfigure}
		\begin{subfigure}{.2\textwidth}
			\centering
			\includegraphics[width=\textwidth]{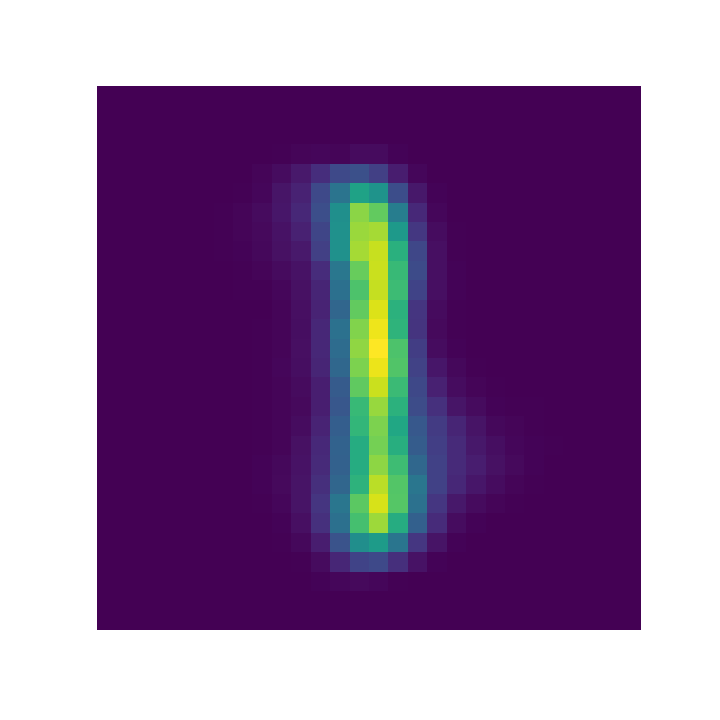}
			\caption*{MMSE}
		\end{subfigure}\\
		\begin{subfigure}{.8\textwidth}
			\hfill\\
			\centering
			\includegraphics[width=\textwidth]{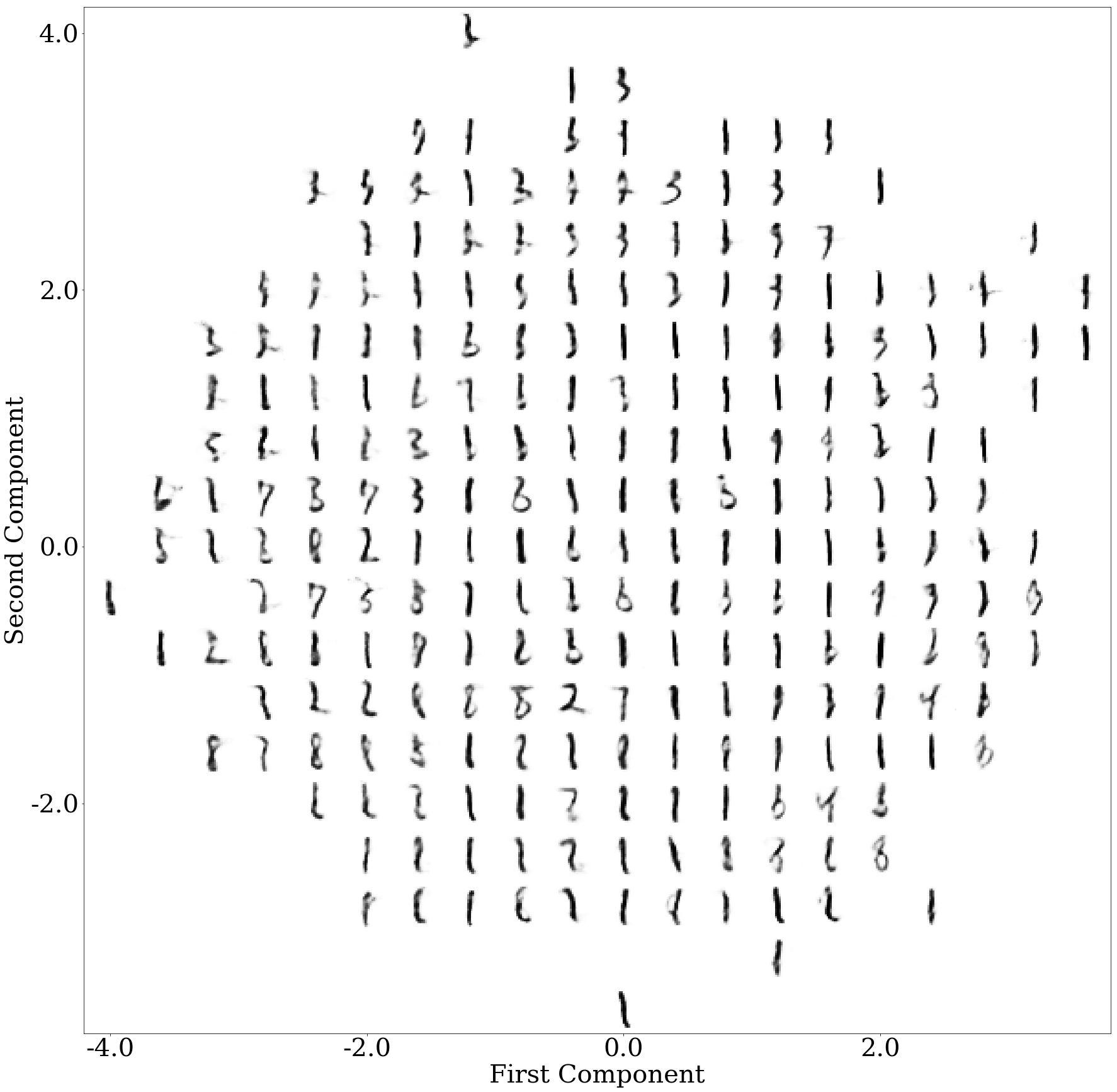}
		\end{subfigure}
	\caption{Map of uncertainty in principal components for deblurring.}
	\label{pca0}
\end{figure}

\subsection{Detecting Out of Dataset Examples} \label{margLik}
When using a data-driven prior it is implicitely assumed that the true data generating the observation is an i.i.d. sample from the same distribution as the training data. This assumption can be tested, in the spirit of Neyman--Pearson, by estimating the marginal likelihood of the observed data under the model.  This is given by
\[
p(y) = \int_{\mathbb{R}^m} p(y|z) p(z) \, dz.
\]
Given an observation, $y$, and $k$ chains of samples at temperatures $\{T_i\}_{i=1}^k$ generated by Algorithm \ref{ParallelpCN}, the marginal likelihood can be estimated \cite{friel2008marginal} by

\[
\log p(y) \approx \sum_{i=1}^{k-1} (T_{i+1} - T_i) \frac{\mathbb{E}_{z|y,T_{i+1}} \{\log p(y|z) \} + \mathbb{E}_{z|y,T_{i}} \{ \log p(y|z) \} }{2},
\]
where $\mathbb{E}_{z|y,T_{i+1}}$ denotes the expectation with respect to the density at temperature $T_i$, given by $p(z|y,T_i) \propto p(y|z)^{T_i} p(z)$.  Having generated samples $\{z_j^{(i)}\}$ at temperatures $T_i$, a Monte Carlo estimate of these expectations can be calculated with little extra cost:
\[
\mathbb{E}_{z|y,T_{i}} \{ \log p(y|z) \} \approx \frac{1}{N} \sum_{j=1}^N \log p(y|z^{(i)}_j).
\]
By estimating the marginal log-likelihood of observations from the training data, a reference distribution of marginal log-likelihood values can be constructed, for comparison with new observations.  The marginal likelihood of a new observation $y$ may then be compared with the reference distribution to test whether the newly observed data are generated from the same distribution, and misspecification can be diagnosed at a given confidence level.  To test at the 5\% significance level, for example, one simply takes as critical value the 5th centile of the  marginal likelihoods obtained from the training set, and reject any new observations with likelihood below the critical value. This is illustrated in Figure \ref{boxplots}, which shows boxplots of the marginal likelihood of the data for in-dataset (MNIST) and out-of-dataset (notMNIST \cite{bulatov2011notmnist}) examples, for each of the forward operators from Section \ref{subsec:ptEstimates} and the lowest noise level ($\sigma = 0.1, 0.01, 0.01$  for denoising, inpainting and deblurring respectively).  Non-parametric likelihood ratio tests at the 1\% level have empirically estimated powers of 99.6\%, 88.5\% and 99.8\% for denoising, inpainting and deblurring respectively when distinguishing between the MNIST and notMNIST datasets.  The distributions are displayed as histograms in Figure \ref{histograms}.

\begin{figure}[t]
	\centering
	\begin{minipage}{0.32\textwidth}
		\centering
		\begin{subfigure}{\textwidth}
			\centering
			\includegraphics[width=\textwidth]{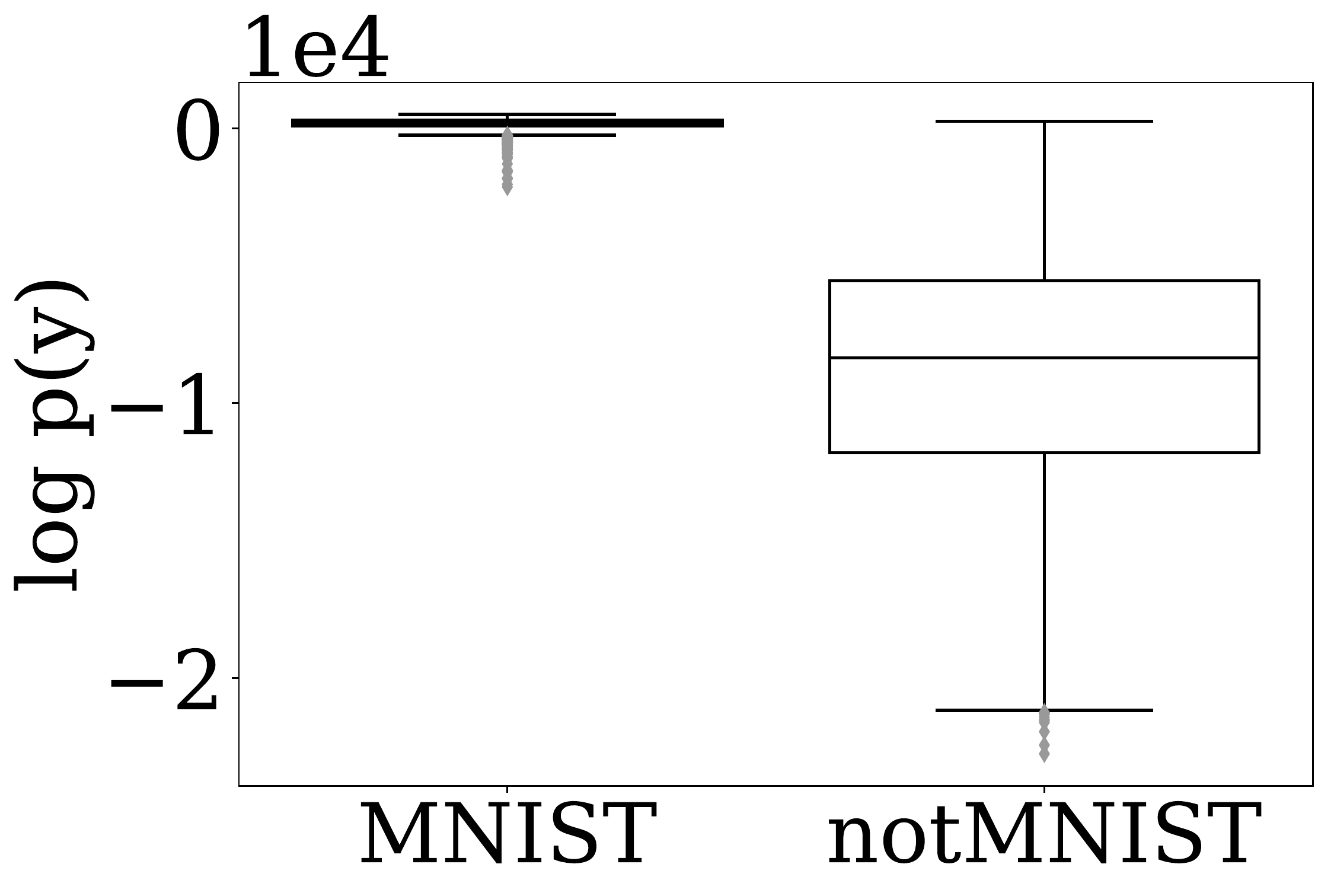}
		\end{subfigure}
		\caption*{Denoising}
	\end{minipage}\hfill
	\begin{minipage}{0.32\textwidth}
		\centering
		\begin{subfigure}{\textwidth}
			\centering
			\includegraphics[width=\textwidth]{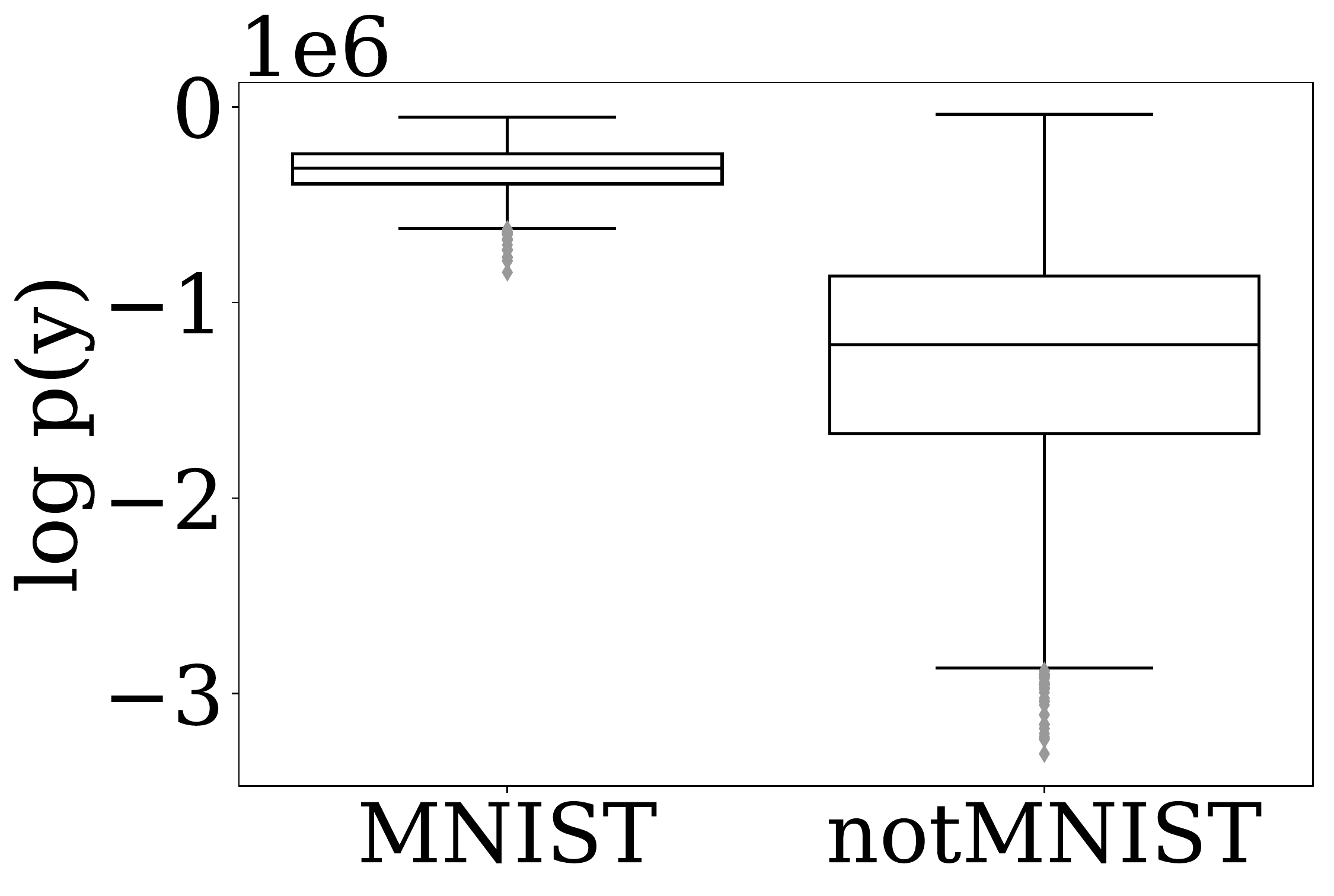}
		\end{subfigure}
		\caption*{Inpainting}
	\end{minipage}\hfill
	\begin{minipage}{0.32\textwidth}
		\centering
		\begin{subfigure}{\textwidth}
			\centering
			\includegraphics[width=\textwidth]{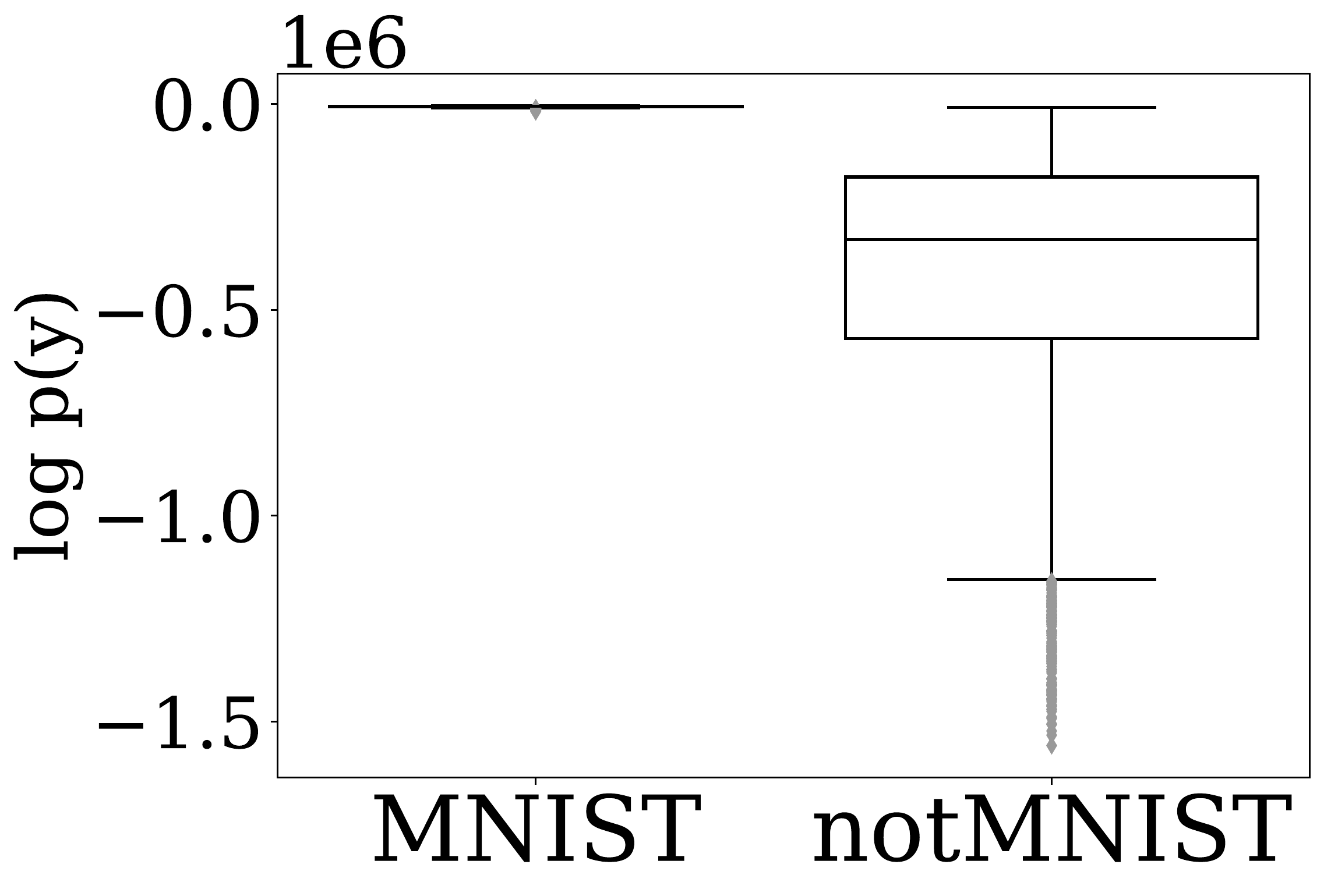}
		\end{subfigure}
		\caption*{Deblurring}
	\end{minipage}
	\caption{Boxplots of marginal likelihood of in- and out-of-dataset observations for 3 forward operators.  The noise levels were set at $\sigma=0.1$, 0.01 and 0.01 for denoising, inpainting and deblurring respectively.  Non-parametric likelihood ratio tests at the 1\% level to distinguish between observations from the MNIST and notMNIST datasets  have empirically estimated powers of  99.6\%, 88.5\% and 99.8\% for denoising, inpainting and deblurring respectively.}
	\label{boxplots}
\end{figure}

\begin{figure}[h!]
	\centering
	\begin{minipage}{0.3\textwidth}
		\centering
		\begin{subfigure}{\textwidth}
			\centering
			\includegraphics[width=\textwidth]{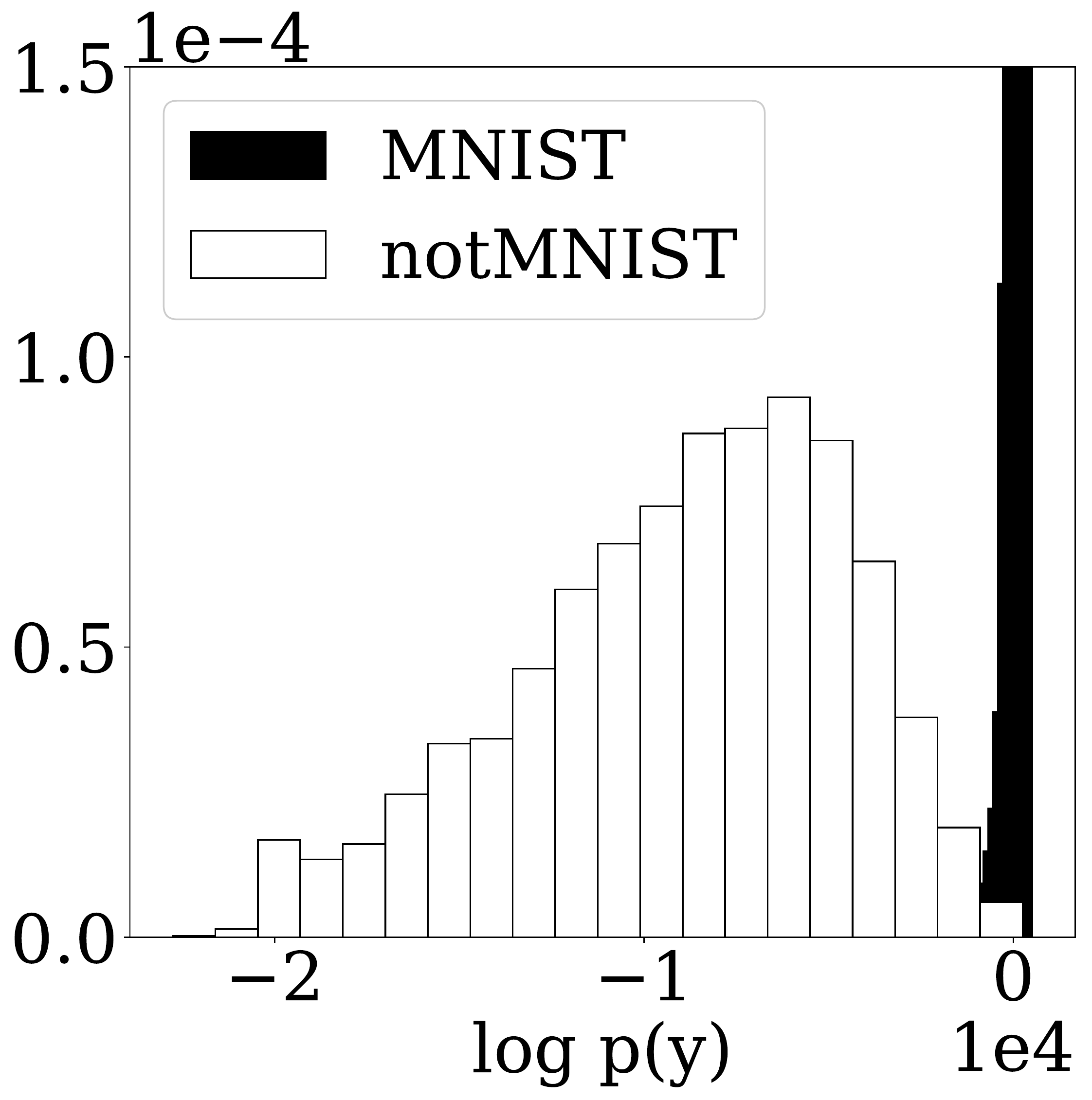}
			\caption{Denoising}
		\end{subfigure}
	\end{minipage}\hfill
	\begin{minipage}{0.3\textwidth}
		\centering
		\begin{subfigure}{\textwidth}
			\centering
			\includegraphics[width=\textwidth]{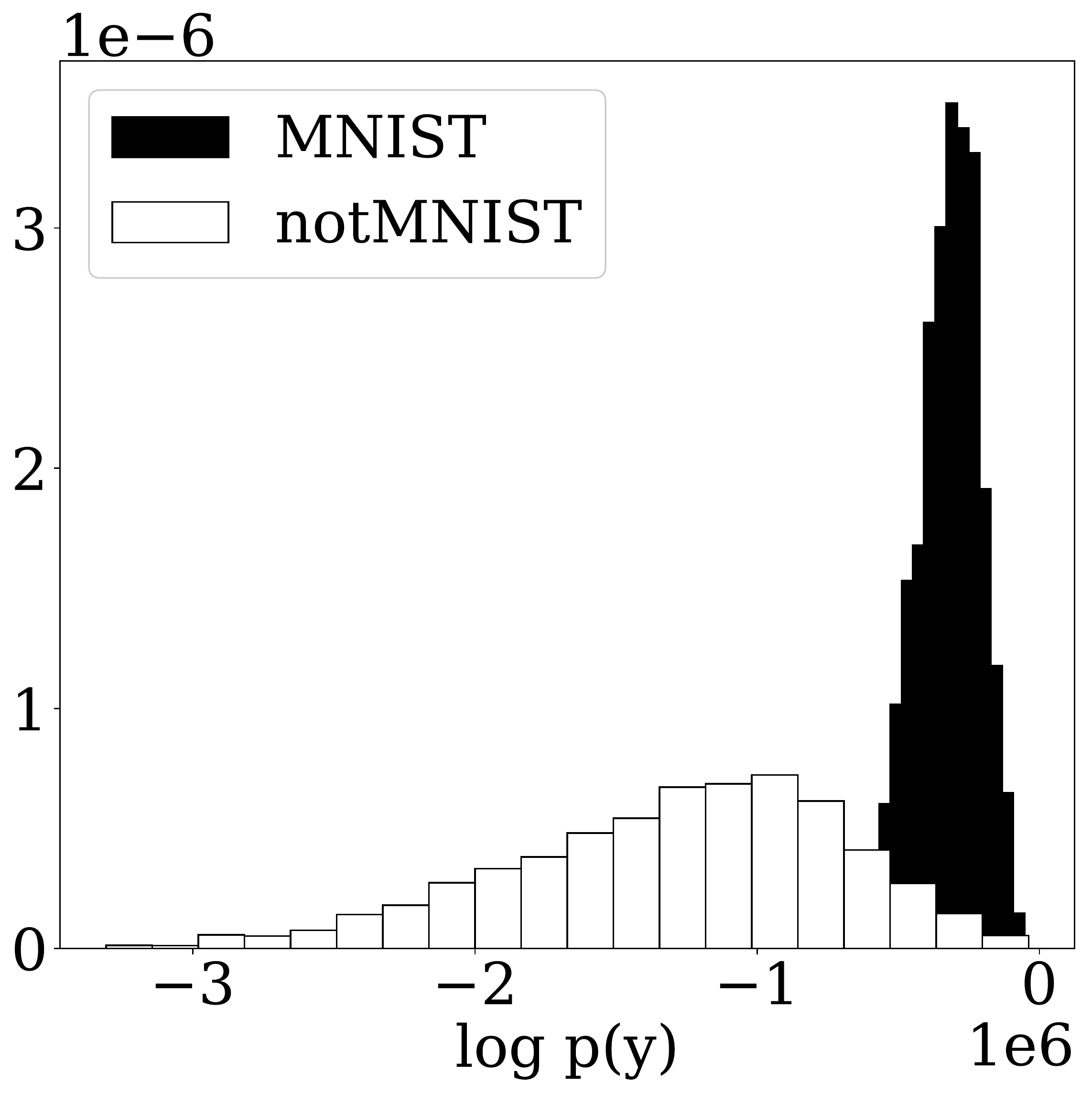}
			\caption{Inpainting}
		\end{subfigure}
	\end{minipage}\hfill
	\begin{minipage}{.3\textwidth}
		\centering
		\hfill
		\begin{subfigure}{\textwidth}
			\centering
			\includegraphics[width=\textwidth]{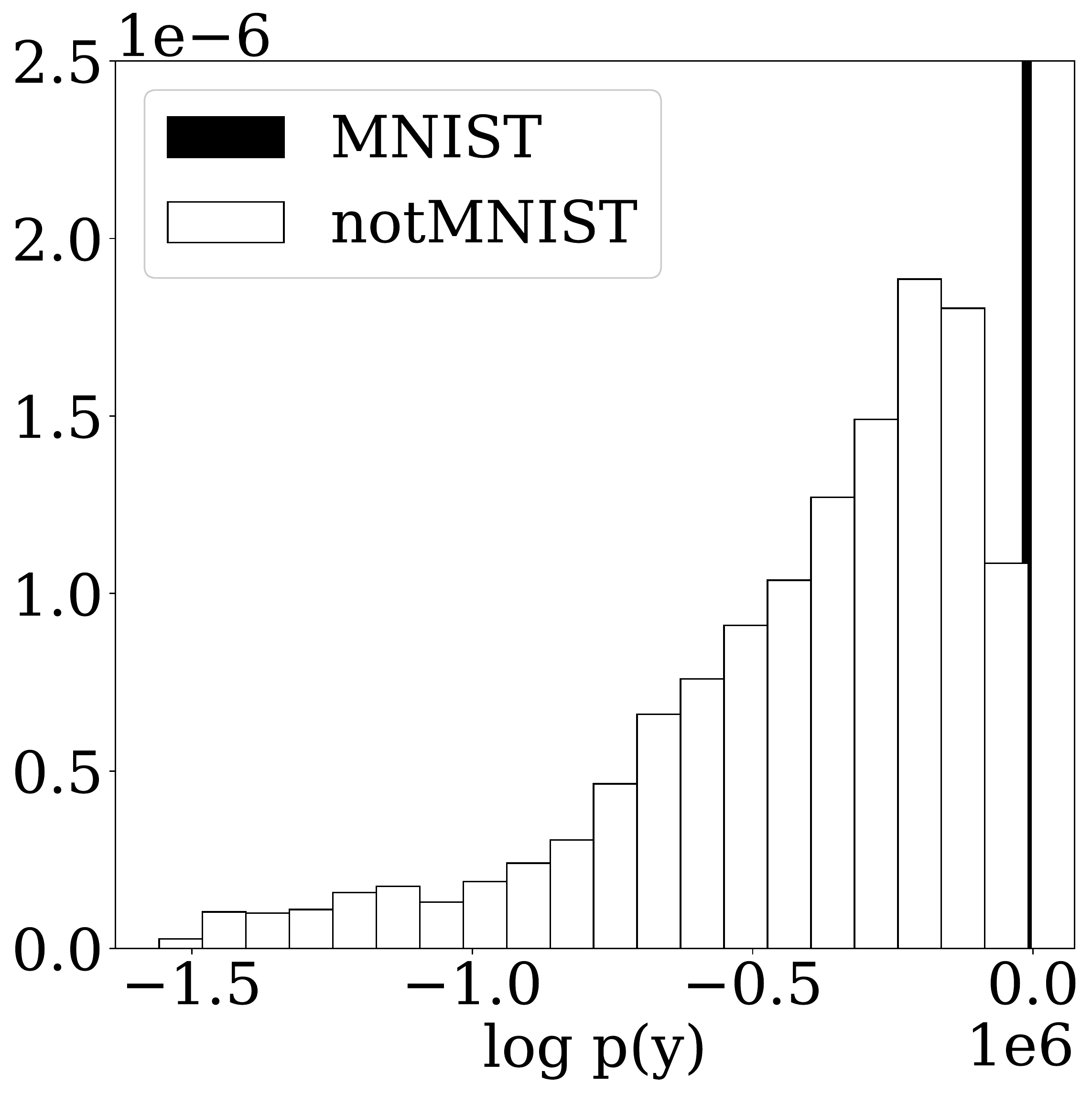}
			\caption{Deblurring}
		\end{subfigure}
	\end{minipage}
	\caption{Histograms of marginal likelihood of in- and out-of-dataset observations for experiments in Figure \ref{boxplots}. For denoising and deblurring, the distribution of in-dataset marginal likelihoods is much more concentrated than out-of-dataset, so the y-axis is restricted to better visualise the out-of-dataset distribution (Figures (a) and (c)).}
	\label{histograms}
\end{figure}

\subsection{Frequentist Coverage Properties} \label{Coverage}

Finally we present an experiment interrogating the quality of the prior in a frequentist manner.  Currently, the priors used in imaging problems result in posterior probabilities which are only true in the subjective sense of the model, rather than a frequentist sense.  Under a frequentist interpretation, the probabilities should match the frequency observed over a large number of trials. If $\{x_1,...,x_n\}$ were observations of the quantity being modelled, then for large n, one should observe that

\[
\frac{1}{n} \sum_{i=1}^n \mathbbm{1}_A(x_i) \approx \mathbb{P}(x \in A) = \int_A p(x|y) p(y) \, dy.
\]

In practice, however, empirical averages do not coincide with the reported posterior probabilities. Data-driven priors may offer a solution to this problem, and could report probabilities which are much closer to the empirical averages observed.  To test the properties of the model being used here, we perform a test of the coverage of the credible intervals reported by the model. This is done as follows:
\begin{enumerate}
	\item Given a set $\{x^\dagger_i\}_{i=1}^n$ from the testing data, generate a corresponding observation $\{y^\dagger_i\} \sim p(y|x=x^\dagger_i)$.
	
	\item For each $y^\dagger_i$, sample the posterior to find credible intervals for different levels $\alpha$, i.e. find a region $A_i$ such that the posterior probability $\mathbb{P}(x \in A | y=y^\dagger_i) = 1 - \alpha$.
	
	\item Observe the proportion of credible intervals which contain the true value $x^\dagger_i$.  If the prior is true in a frequentist sense, this should be close to $1-\alpha$.
\end{enumerate}

In this case, since the posterior distribution is defined on the latent space, but the true quantity of interest is the image $x$ in the pixel space, it is difficult to construct such a set $A$ from the samples generated.  However, the encoder of the VAE can be used to conduct an approximate test.  For a set of true data $\{x^\dagger_i\}_{i=1}^n$ in the test set, we generate observations $\{y^\dagger_i\}_{i=1}^n$ according to the forward model, and sample the posterior as outlined in Section \ref{subsec:posteriorlatent}.  From these samples, we wish to generate non-parametric credible intervals (highest posterior density intervals) using quantiles of the posterior density, and compare the empirical proportion of $\{x^\dagger_i\}$ lying within the credible interval to the theoretical probability.  By passing $x^\dagger_i$ through the VAE encoder, we can generate $z_i \sim q_\phi(z|x^\dagger_i)$ such that $x^\dagger_i\approx \mu_\theta(z_i)$, and likewise for the posterior samples.  The coverage test is then performed by evaluating the posterior at these encoded values.  An error is introduced in this process, since the encoder is not an exact inverse of the decoder,  but the coverage nevertheless gives an indication of the suitability of the prior.

The coverage is tested for differing noise levels and forward operators (Figure \ref{coverage}).  For denoising, the empirically observed coverage is remarkably close to the theoretical values, indicating that the posterior probabilities reported are accurate in a frequentist sense. For large levels of noise, the actual coverage slightly exceeds the theoretical, indicating that the models are mildly conservative. Similar results also exhibited in the large-noise level cases of inpainting and deblurring.  As one would expect, the prior misspecification has a larger effect on the ill-posed problems of inpainting and deblurring, and the credible intervals contain the encoded true image with a higher frequency than the probability reported. In any case, in all the considered problems, the difference between the theoretical and the empirical coverage is small enough to be compensated by calibration (i.e., by using empirically estimated thresholds from Figure \ref{coverage} to specify the highest posterior density), particularly so for the range of values that is most relevant for Bayesian hypothesis testing (e.g, from 80\% and 99\%). It is worth emphasising at this point that, to the best of our knowledge, this is the first example of a Bayesian imaging model that has accurate frequentist properties, albeit for a simple dataset. 

Moreover, since the empirical coverage results rely on accurate sampling of the posterior, the inaccurate coverage may partly be a result of sampling error, particularly in the case of deblurring, where the ill-conditioning makes full exploration of the posterior more difficult.  The inaccurate coverage for low noise levels supports this idea, since the mass of the posterior is more concentrated, and exploration of the full posterior more difficult as a result. We therefore expect that the coverage should be even better when we adopt a gradient-based variant of pCN that is currently under development.

\begin{figure}[t]
	\centering
	\begin{subfigure}{.3\textwidth}
		\centering
		\includegraphics[width=\textwidth]{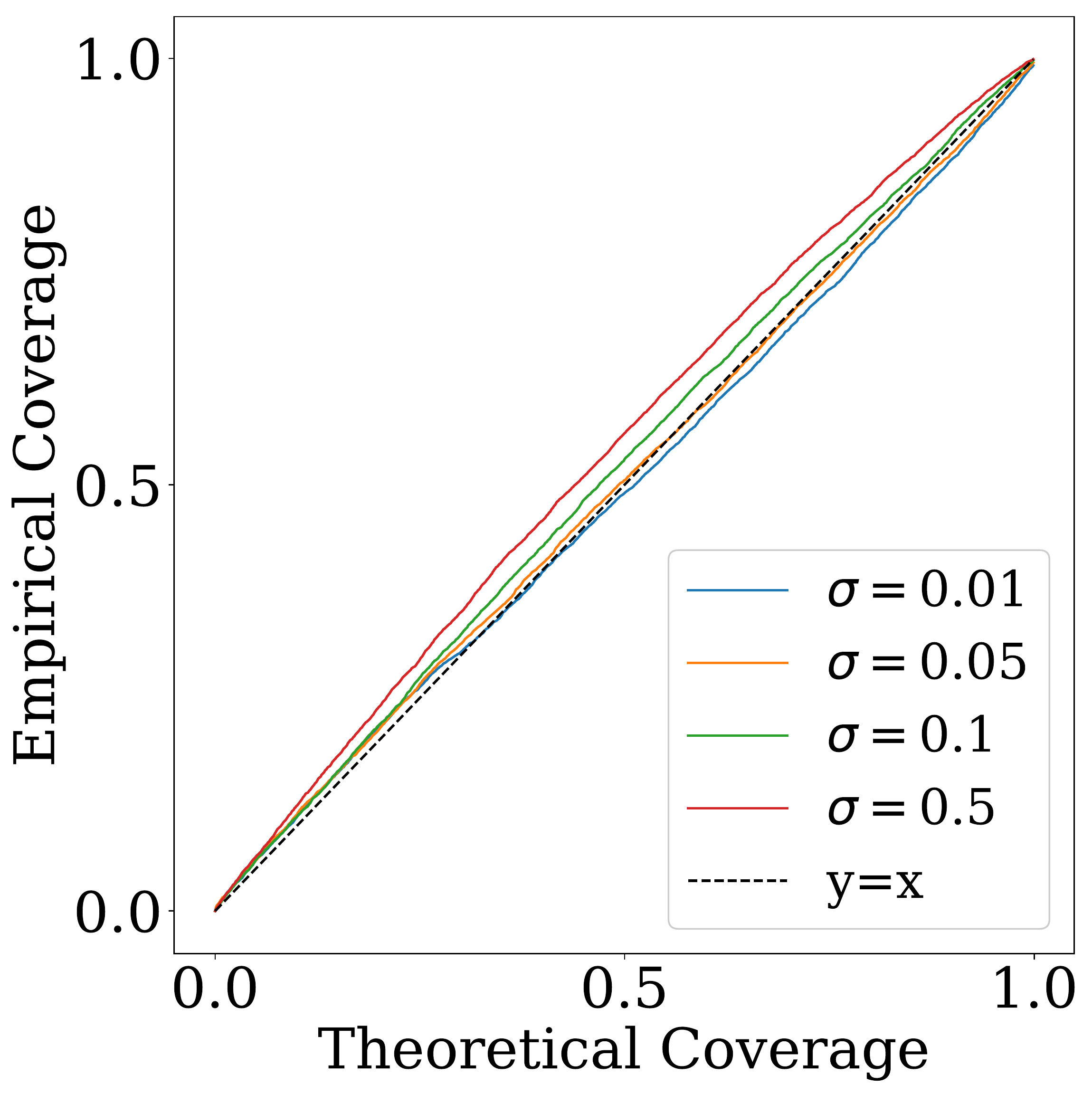}
		\caption{Denoising}
	\end{subfigure}
	\begin{subfigure}{.3\textwidth}
		\centering
		\includegraphics[width=\textwidth]{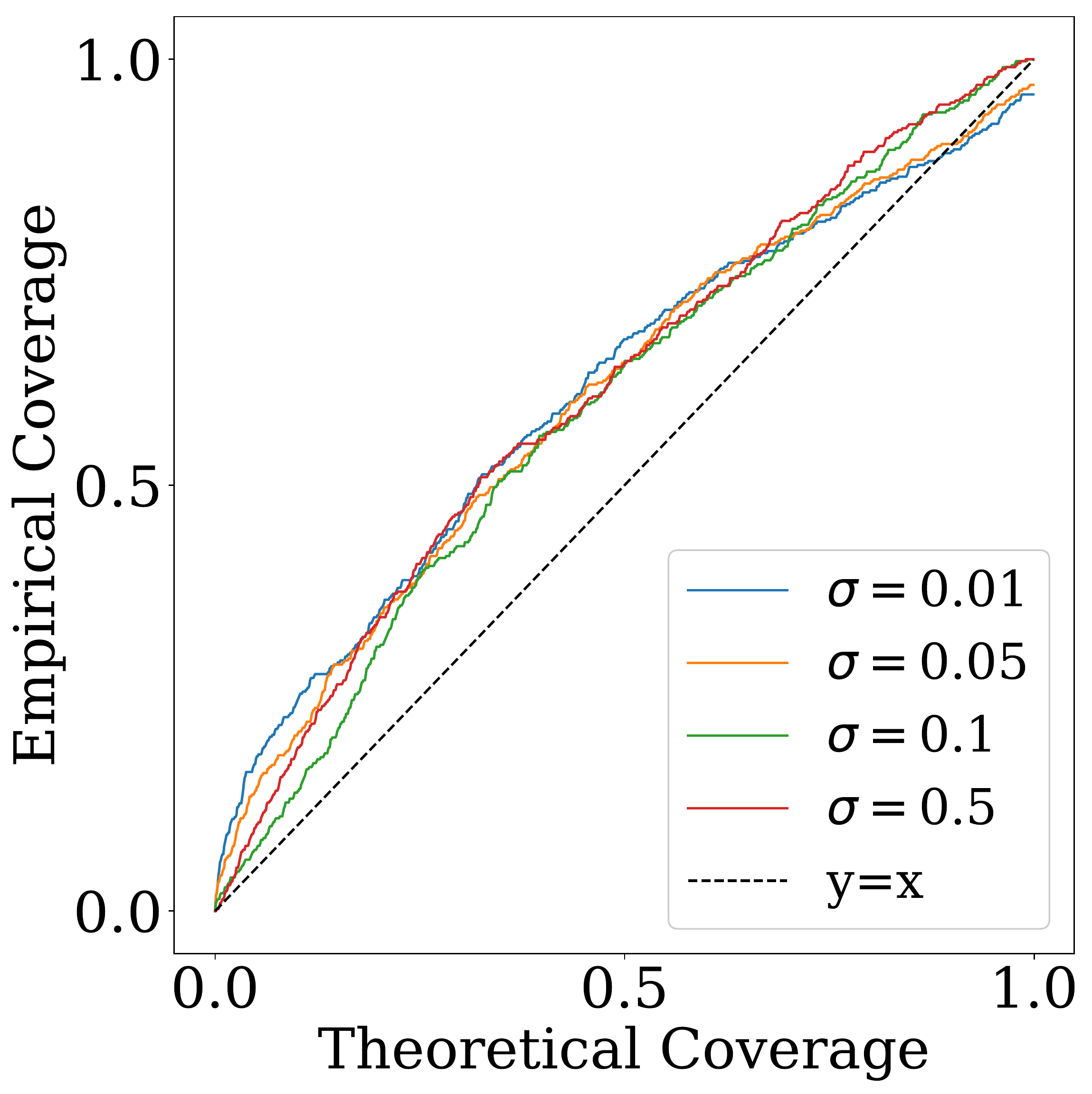}
		\caption{Inpainting}
	\end{subfigure}
	\begin{subfigure}{.3\textwidth}
		\centering
		\includegraphics[width=\textwidth]{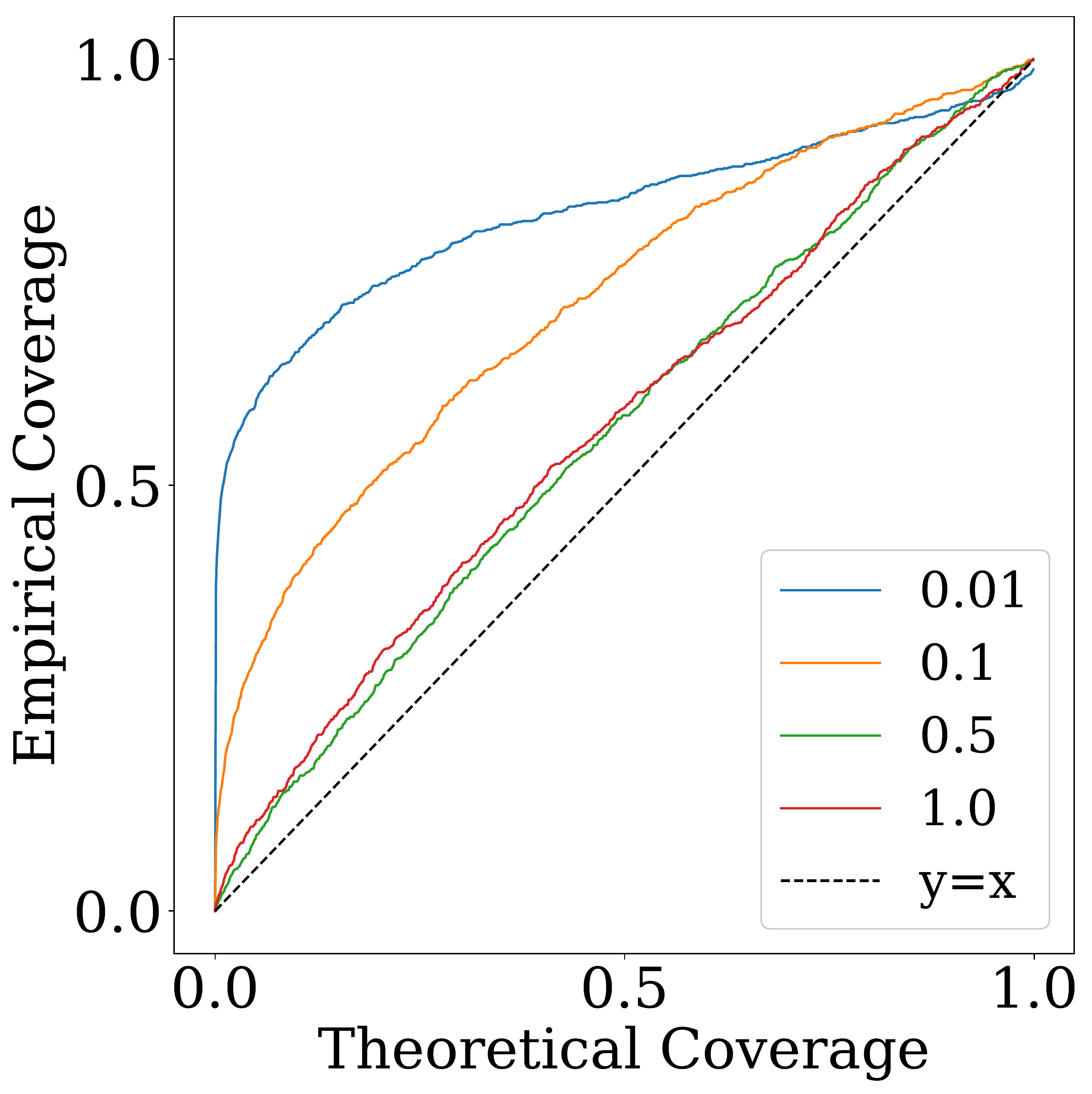}
		\caption{Deblurring}
	\end{subfigure}
	\caption{Coverage Tests.}
	\label{coverage}
\end{figure}

\section{Conclusion} \label{sec:discus}
This work presented a new Bayesian methodology for solving inverse problems in imaging sciences, with a focus on applications with training data available. The methodology is based on a data-driven prior that is supported on a manifold, which is encoded by a neural network and learnt from the training data. We established that the posterior distribution associated with this prior is well-posed in the sense of Hadamard under easily verifiable conditions, which are in particular verified by the widely used Gaussian linear forward model. Similarly, posterior moments are also guaranteed to exist and to be well-posed, providing a rigorous underpinning for Bayesian estimators such as the posterior mean. In order to perform Bayesian computation with the proposed data-driven prior, we constructed a parallel tempered pCN Monte Carlo algorithm on the low-dimensional latent space that is proven to be ergodic and robust to multi-modality. This then allowed computing the MMSE estimator as well as performing uncertainty quantification and model misspecification analyses.  The performance of the method was illustrated with an application to the MNIST dataset for the tasks of image denoising, inpainting and deconvolution, and it was observed that the MMSE estimates were more robust to noise and poor conditioning than competing approaches from the state of the art based on MAP estimation.  In addition, we demonstrated how to detect observations generated from out-of-dataset images, thus tackling one of the main risks of using a data-specific prior.  The frequentist properties of our model were also examined and found to be remarkably accurately. To the best of our knowledge, this is the first example of a Bayesian model with accurate frequentist coverage properties in an imaging setting. We also presented a PCA-based strategy for visualising the uncertainty in the solutions delivered and discussed how to infer the dimension of the latent manifold directly from the training data. 

A main perspective for future work is to scale the proposed methodology to large images for example by exploring the use of gradient-based MCMC algorithms \cite{titsias2018auxiliary}, by developing a patch-based approach \cite{Zoran2011,houdard2018high, teodoro2020block}, and by replacing the standard VAE approach with a generative model that is well-posed in infinite dimensions \cite{dupont2021generative}. An alternative, but related, line of research would be to consider a model where the prior is defined not on a manifold defined by a generative model, but rather by the fixed point set of a denoising algorithm, as in \cite{cohen2020regularization}.  That work takes a variational approach to the problem, but it would be interesting to extend this rigorously within a Bayesian framework. Future work could also study the application of the proposed methodology to other imaging problems, particularly problems that are severely ill-posed and that involve highly non-regular models as in low-photon imaging.

\section*{Acknowledgments}
The authors are grateful for useful discussions with Andr\'{e}s Almansa.

\end{document}